\documentclass[11pt]{amsart}
\usepackage{geometry}                
\usepackage{graphicx}
\usepackage{amssymb}
\usepackage{epstopdf}
\usepackage{graphicx}
\usepackage{amsmath}
\usepackage{mathrsfs}
\usepackage{caption}
\usepackage{subcaption}
\usepackage{color}
\usepackage{enumerate}

\usepackage[normalem]{ulem}

\title{Incommensurate Heterostructures in Momentum Space}

\author[D. Massatt]{Daniel Massatt}
\address{D. Massatt \\ School of Mathematics \\ University of Minnesota \\ Minneapolis, Minnesota, 55455 \\ USA.}
\email{massa067@umn.edu}

\author[S. Carr]{Stephen Carr}
\address{S. Carr \\ Department of Physics \\ Harvard University \\ Cambridge, Massachusetts 02138 \\ USA}
\email{stephencarr@g.harvard.edu}

\author[M. Luskin]{Mitchell Luskin}
\address{M. Luskin \\ School of Mathematics \\ University of Minnesota \\ Minneapolis, Minnesota, 55455 \\ USA}
\email{luskin@umn.edu}

\author[C. Ortner]{Christoph Ortner}
\address{C. Ortner \\ Mathematics Institute \\ University of Warwick \\ Coventry CV4 7AL \\ UK}
\email{c.ortner@warwick.ac.uk}

\thanks{DM was supported by NSF PIRE Grant OISE-0967140 and ARO MURI Award W911NF-14-1-0247. SC and ML were supported in
  part by ARO MURI Award W911NF-14-1-0247. CO was supported by ERC Starting Grant
  335120.}

\date{\today}                                           

\numberwithin{equation}{section}

\begin{document}

\newcommand{\ml}[1]{{\color{red} #1}}
\newcommand{\cml}[1]{{\small \it \color{red} [ML: #1]}}
\newcommand{\co}[1]{{\color{blue} #1}}
\newcommand{\cco}[1]{{\small \it \color{blue} [CO: #1]}}
\newcommand{\dm}[1]{{\color{magenta} #1}}
\newcommand{\stc}[1]{{\color{green} #1}}

\def\Xint#1{\mathchoice
{\XXint\displaystyle\textstyle{#1}}%
{\XXint\textstyle\scriptstyle{#1}}%
{\XXint\scriptstyle\scriptscriptstyle{#1}}%
{\XXint\scriptscriptstyle\scriptscriptstyle{#1}}%
\!\int}
\def\XXint#1#2#3{{\setbox0=\hbox{$#1{#2#3}{\int}$ }
\vcenter{\hbox{$#2#3$ }}\kern-.6\wd0}}
\def\mint{\Xint-}

 \newtheorem{assumption}{Assumption}
\newtheorem{remark}{Remark}
\newtheorem{prop}{Proposition}
\newtheorem{thm}{Theorem}
\newtheorem{lemma}{Lemma}
\newtheorem{definition}{Definition}
\newtheorem{corollary}{Corollary}
\numberwithin{definition}{section}
\numberwithin{thm}{section}
\numberwithin{remark}{section}
\numberwithin{prop}{section}
\numberwithin{corollary}{section}
\numberwithin{assumption}{section}
\numberwithin{lemma}{section}

\newcommand{\tbeta}{\tilde \beta}
\newcommand{\oJ}{\overline{J}}
\newcommand{\oI}{\overline{I}}
\newcommand{\oP}{\overline{P}}
\newcommand{\wH}{\widehat{H}}
\newcommand{\Id}{I}
\newcommand{\R}{\mathcal{R}}
\newcommand{\I}{\mathcal{I}}
\newcommand{\Tr}{\text{Tr}}
\newcommand{\TrN}{\text{Tr}_N}
\newcommand{\adj}{\text{adj}}
\newcommand{\per}{\text{per}}
\newcommand{\C}{\mathcal{C}}
\newcommand{\J}{\mathcal{J}}
\newcommand{\interior}{\text{int}}
\newcommand{\sch}{\mathcal{S}(\mathbb{R})}
\newcommand{\supp}{\text{supp}}
\newcommand{\Err}{\text{Err}}
\newcommand{\Imag}{\text{Im}}
\newcommand{\Real}{\text{Re}}
\newcommand{\rins}{r_{\text{ins}}}
\newcommand{\hQ}{Q}
\newcommand{\E}{\mathcal{E}}
\newcommand{\Mat}{\tilde H}
\newcommand{\G}{\mathcal{G}}
\newcommand{\Gt}{\widetilde{\mathcal{G}}}
\newcommand{\dhh}{\widehat{\delta h}}
\newcommand{\Z}{\mathbb{Z}^2}
\newcommand{\K}{\mathcal{R}^*} 
\newcommand{\tK}{{\tilde \R^*}}
\newcommand{\Br}{B_r(0)}
\newcommand{\B}{\mathcal{B}}
\newcommand{\A}{\mathcal{A}}
\newcommand{\yt}{\tilde{y}}
\newcommand{\Rs}{\mathscr{R}}
\newcommand{\Ha}{\mathcal{H}}
\newcommand{\SN}{\mathcal{S}_N}
\newcommand{\gap}{\text{gap}}
\newcommand{\inter}{\text{inter}}
\newcommand{\intra}{\text{intra}}
\newcommand{\Gammat}{\tilde{\Gamma}}
\newcommand{\OmegaMon}{\Omega}
\newcommand{\D}{\mathcal{D}}
\newcommand{\Hmon}{H}
\newcommand{\MatSpace}{ M_{|\Omega_r|}(\mathbb{C})}
\newcommand{\M}{\mathcal{M}}
\newcommand{\Msup}{S[\widehat{H}]}
\newcommand{\nmod}{\text{mod}}
\newcommand{\hOmega}{\widehat{\Omega}}
\newcommand{\T}{\mathcal{T}}
\newcommand{\U}{\mathcal{U}}
\newcommand{\rc}{r_c}
\newcommand{\V}{\mathcal{V}}
\newcommand{\moire}{\theta}
\newcommand{\mP}{\mathcal{P}}

\begin{abstract}
To make the investigation of electronic structure of incommensurate heterostructures computationally tractable, effective alternatives to Bloch theory must be developed. In \cite{massatt2017}, we developed and analyzed a real space scheme that exploits spatial ergodicity and near-sightedness. In the present work, we present an analogous scheme formulated in momentum space, which we prove have significant computational advantages in specific incommensurate systems of physical interest, e.g., bilayers {of a specified class of materials} with small rotation angles. We use our theoretical analysis to obtain estimates for improved rates of convergence with respect to total CPU time for our momentum space method that are confirmed in computational experiments.
\end{abstract}
\maketitle

\begin{section}{Introduction}

The electronic structure of 2D heterostructures is typically studied using Bloch theory. This method transforms the system into the Bloch basis, or into momentum space, reducing the problem size to that of the periodic cell \cite{kaxiras2003}. This approach breaks down for {\em incommensurate} bilayer systems where there is no periodic cell \cite{Terrones2014,2DPerturb15}. In \cite{massatt2017, carr2017} we introduced a method for calculating the Density of States (DoS) without using supercell approximations via locality of the Hamiltonian and equidistribution of the site configurations in real space.

In certain settings, the Bloch bases of the monolayers can still be used to approximate electronic properties such as the density of states when the two layers are similar in size. In \cite{weckbecker2016}, k.p theory uses the wave function basis in a tight-binding setting to get an approximate Hamiltonian accurate for an energy range around the graphene Dirac point. In \cite{koshino2015}, a locality method is introduced to calculate the Density of States in a momentum space formulation parallel to what was done in \cite{massatt2017} for real space.

In previous work it is left unclear for which materials and orientations the momentum space method is effective. For example it is suggested in \cite{koshino2015} that it is effective for arbitrary materials and relative orientations. Moreover, they use a standard circular local matrix approximation, which we demonstrate is sub-optimal.
In the present work, by establishing rigorous convergence we are able to derive (quasi-)optimal choices of local matrices. We also show how one can deduce, from the monolayer band structure alone, for what energies and materials the method is effective compared with the real space method {(Section \ref{sec:approx:region})}. For example, we show that at the famous bilayer graphene Dirac energy the momentum space method significantly outperforms the real space method, while at energies corresponding to flat bands, the momentum space method has no noticeable advantage. Moreover, we will see that the momentum space method is particularly effective when the {underlying Bravais lattices are ``very close'', e.g.
 rotated copies of one another with very small rotation angle.}
  This situation encompasses for example many interesting moir\'e problems {\cite{Dai2016,bistritzer2011}.}

We demonstrate the method is effective for energies where the monolayer band structure's corresponding level sets do not wrap in a ring around the Brillouin zone torus. It is shown that in the correct energy regions, this method converges asymptotically faster than the real space method.  The method also promises to be even more efficient for calculating more complicated electronic observables such as conductivity. However, it is not clear if this method can be extended to any more complicated systems than the homogenous incommensurate bilayers. For example, introducing atomic relaxation or defect distributions breaks the symmetry in the individual sheets, and the momentum space analysis breaks down.

{\it Outline: } In Section \ref{sec:incommensurate}, we introduce the momentum space formulation applied to the incommensurate bilayer system. In Section \ref{sec:approx}, we establish the convergence results. In Section \ref{sec:numerics}, we verify the convergence results for twisted bilayer graphene and the equivalence of the momentum space and real space formulations. In Section \ref{sec:proofs}, we give proofs for our results.

\end{section}

\begin{section}{Incommensurate Bilayer}
\label{sec:incommensurate}
\subsection{Bloch Theory for a Monolayer}
\label{sub:mon}
We begin by reviewing how Bloch theory can be employed to efficiently obtain the DoS in a periodic system.

Given a Bravais lattice
\begin{equation*}
\R = \{ A n \text{ : } n \in \mathbb{Z}^2\},
\end{equation*}
where $A$ is a $2\times2$ invertible matrix, we define its associated unit cell
\begin{equation*}
\Gamma := \{ A \beta : \beta \in [0,1)^2\}.
\end{equation*}
The reciprocal lattice and associated reciprocal lattice unit cell are,
respectively, given by
\begin{align*}
   \K &:= \{ 2\pi(A^{-T}) n : n \in \mathbb{Z}^2\},\\
   \Gamma^* &:= \{2 \pi (A^{-T}) \beta : \beta \in [0,1)^2\}.
\end{align*}

Let $\A$ be a set of orbital indices. Then the degree of freedom space
for the monolayer is 
\begin{equation*}
\OmegaMon = \R \times \A.
\end{equation*}
A (real space) Hamiltonian can be defined via the infinite matrix
\begin{equation}
   \label{def:Hmon}
   \Hmon_{R\alpha,R'\alpha'} = h_{\alpha\alpha'}(R-R')
\end{equation}
where $h_{\alpha\alpha'} : \R \rightarrow \mathbb{C}$, $h_{\alpha\alpha'} \in \ell^1(\R)$.
For $\psi \in \ell^2(\R \times \A)$, we define intralayer coupling
\begin{equation*}
[\Hmon\psi]_{\alpha}(R) = [h \ast \psi]_\alpha(R) := \sum_{\alpha' \in \A}\sum_{R' \in \R}
      h_{\alpha\alpha'}(R-R') \psi_{\alpha'}(R').
\end{equation*}

We are interested in electronic properties, which are dependent on the spectrum of the Hamiltonian. For an infinite matrix this is typically challenging. However, Bloch theory yields a convenient solution for the periodic system. For any index set $\I$ (e.g., $\I = \A$), we define the Bloch operator $\G : \ell^2(\R\times \I) \rightarrow C_{\text{per}}(\Gamma^*)\times \I$ by
\begin{equation*}
[\G \psi]_\alpha(q) = \sum_{R \in \R} \psi_{\alpha}(R) e^{-iq\cdot R}, \hspace{3mm} \alpha \in \I.
\end{equation*}
Note that for notation purposes we assume functions in $C_{\text{per}}(\Gamma^*)$ have periodic extensions so they can be defined over $\mathbb{R}^2$.
We will use $\I = \A$ to transform orbitals and $\I = \A^2$ to transform $h = (h_{\alpha\alpha'})_{\alpha\alpha' \in \A}$.

We apply the Bloch transform to $\Hmon \psi$ for $\psi \in \ell^2(\R \times \A)$ to get
\begin{equation}
\label{e:bloch}
\begin{split}
   [\G \Hmon \psi](q) &= \G [h \ast \psi](q) = [\G h ](q)[\G \psi](q).
\end{split}
\end{equation}
We can then rewrite the eigenproblem $\Hmon \psi = \epsilon \psi$ as
\begin{equation*}
[\G h](q) [\G \psi](q) = \epsilon_{q} [\G \psi](q),
\end{equation*}
where $[\G h](q)$ is an $|\A|\times|\A|$ self-adjoint matrix. The eigenstates can then be visualised as functions of $q$ (band structure);
see Figure \ref{fig:band_structure}.
\begin{figure}[ht]
\centering
\includegraphics[width=.6\linewidth]{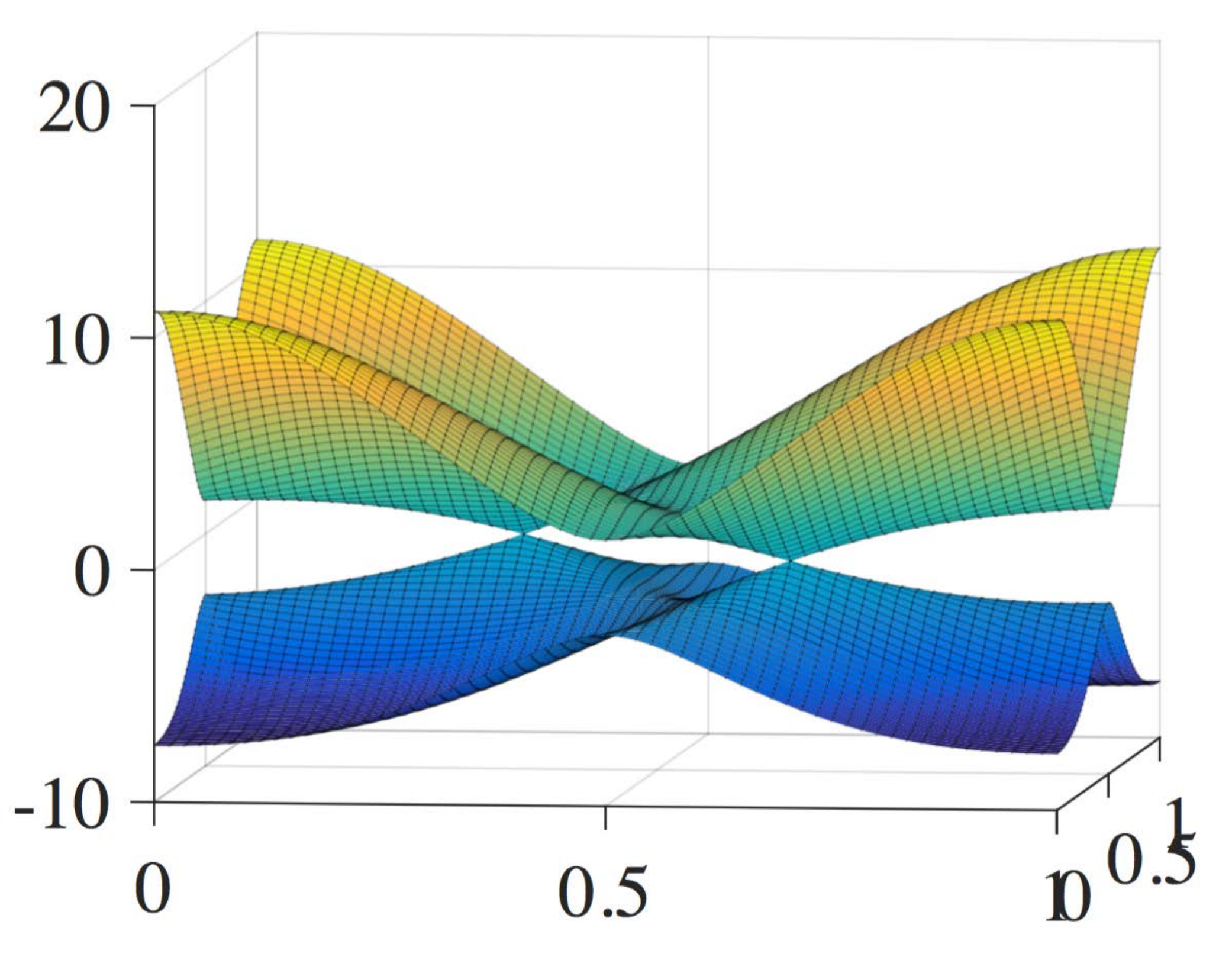}
\caption{Band structure for monolayer graphene using the tight-binding model from \cite{shiang2016}. Here there are only two orbitals because only the two $p_z$ orbitals are relevant for considering electronic properties of graphene near the Fermi level. We are plotting the eigenvalues of $\G h(2\pi A^{-T}y)$ over $ y \in [0,1)^2$.}
   \label{fig:band_structure}
\end{figure}

A natural ``weak'' definition of the density of states
   (heuristically, the distribution of eigenvalues) is
\begin{equation}
   \label{eq:mon_dos_limit}
   \D[H](g) := \lim_{r \rightarrow \infty} \frac{1}{|\Br| \cdot |\A|} \Tr_{\Br} g(H),
\end{equation}
where $\Br$ is the ball centered at the origin with radius $r$ and
$\Tr_{\Br}$ denotes the trace of the $\Omega \cap (\Br \times \A)$-submatrix.
Since $H$ is circulant, we have $[g(H)]_{R\alpha,R\alpha} = [g(H)]_{R'\alpha,R'\alpha}$
for $R,R' \in \R$ and hence we obtain
\begin{equation}
   \label{e:mon_dos}
   \D[H](g) = \frac{1}{|\A|} \sum_{\alpha \in \A}[g(H)]_{0\alpha,0\alpha}.
\end{equation}

Using the inverse operation of the Bloch operator,
\begin{equation*}
   \psi(R) = \mint_{\Gamma^*} [\G \psi](q)e^{iq\cdot R}dq,
\end{equation*}
we can obtain an equivalent momentum space expression for the DoS,
\begin{equation}
   \label{eq:Dhat_monolayer}
   \widehat{\D}[\Hmon](g) :=   \frac{1}{|\A| \cdot |\Gamma^*|}\sum_{\alpha \in \A}\int_{\Gamma^*}[g\circ[\G h](q)]_{\alpha\alpha}dq.
\end{equation}

Let $\mathbb{P}$ be the space of polynomials.

\begin{prop}
   Let $\Hmon$ be defined by \eqref{def:Hmon} with $h \in \ell^1(\R)$. Then
   \begin{equation*}
      \D[\Hmon](g) = \widehat{\D}[\Hmon](g) \qquad
      \forall\, g \in \mathbb{P}.
   \end{equation*}
\end{prop}
\begin{proof}
   It suffices to prove this for $g(x) = x^n$.
   Let $\hat{e}^\alpha$ be defined over $\A$ such that $\hat e^\alpha_\beta = \delta_{\alpha\beta}$, and $e^\alpha : \R \to \mathbb{R}^\A$ given by
   \begin{equation*}
   [e^\alpha]_{\alpha'}(R') = \delta_{0R'}\delta_{\alpha\alpha'},
   \qquad \text{ for } R' \in \R, \quad \alpha, \alpha' \in \A.
   \end{equation*}
    We note that
   \begin{equation*}
      [\G e^\alpha](q) = \sum_R e^{-i q \cdot R}[ e^\alpha](R) = \hat e^\alpha,
   \end{equation*}
   hence we can deduce
   \begin{equation*}
   \begin{split}
   [e^\alpha]^*g(H)e^\alpha& = [\hat e^\alpha]^*\mint_{\Gamma^*} \G \bigl(g(H)e^\alpha\bigr)dq\\
   &= [\hat e^\alpha]^*\mint_{\Gamma^*} [\G (h \ast h \ast  \cdots \ast h \ast e^\alpha)] dq\\
   &=  [\hat e^\alpha]^* \mint_{\Gamma^*} [\G h(q)]^n \hat e^\alpha dq.
   \end{split}
   \end{equation*}
   Inserting this identity into \eqref{e:mon_dos} yields the desired result.
\end{proof}

Equations \eqref{e:mon_dos} and \eqref{eq:Dhat_monolayer} give, respectively,
the real space and momentum space definitions of the DoS in the simplest case of
a mono-layer Bravais lattice.

\subsection{Incommensurate Bilayer}
Before turning to a momentum space formulation for an incommensurate bilayer system,
we first review the representation formula for the DoS from \cite{massatt2017}.
We define the lattices for the two sheets by
\begin{equation*}
\R_j = \{ A_jn \text{ : } n \in \mathbb{Z}^2\},
\end{equation*}
where $A_j$ is a $2\times2$ invertible matrix. We define the {\em unit cell} for sheet $j$ as
\begin{equation*}
\Gamma_j = \{ A_j \beta \text{ : } \beta \in [0,1)^2\}.
\end{equation*}
\begin{definition}
Two Bravais lattices $\mathcal{L}_1$ and $\mathcal{L}_2$ are {\em incommensurate} if, for $v \in \mathbb{R}^2$,
\begin{equation*}
\mathcal{L}_1 \cup \mathcal{L}_2 + v = \mathcal{L}_1 \cup \mathcal{L}_2 \quad \Leftrightarrow \quad  v = \begin{pmatrix} 0\\0 \end{pmatrix}.
\end{equation*}
\end{definition}
\begin{assumption}
   \label{assump:incomm}
   The lattices $\R_1$ and $\R_2$ are incommensurate, and the reciprocal lattices $\K_1$ and $\K_2$ are incommensurate.
\end{assumption}
\begin{figure}[ht]
\centering
\includegraphics{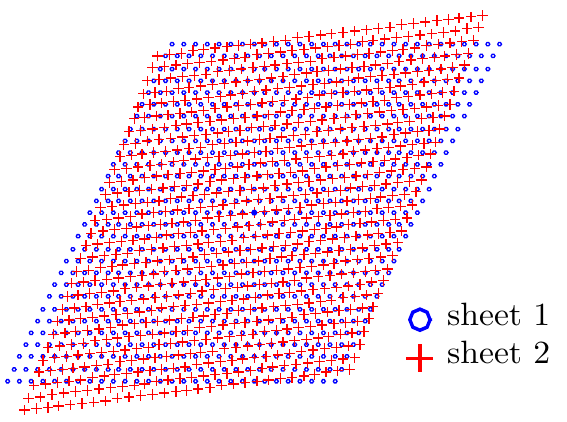}
\caption{An incommensurate hexagonal bilayer.  Sheet 1 is rotated by $\theta = 6^\circ$ relative to sheet 2.}
\label{fig:incommMoire}
\end{figure}

We briefly recall the formalism for the tight-binding model for bilayers \cite{massatt2017, shiang2016}. Let $\A_i$ denote the set of indices of orbitals associated with each unit cell
of sheet $i$. We assume that $\A_i$ are finite and that $\A_1 \cap \A_2 =
\emptyset$. We let $\Omega_j = \R_j \times \A_j$. Then
 the full degree of freedom space is
\begin{equation*}
   \Omega = \Omega_1 \cup \Omega_2.
\end{equation*}
The interaction between orbitals indexed by $R\alpha$ and $R'\alpha'$ is denoted by $h_{\alpha\alpha'}(R-R')$, where $h_{\alpha\alpha'} : \mathbb{R}^2 \rightarrow \mathbb{C}$ if $\alpha$ and $\alpha'$ are not from the same sheet. Recall $h_{\alpha\alpha'} :\R_j \rightarrow \mathbb{C}$ if $\alpha, \alpha' \in \A_j$. Although the sheets have a vertical displacement between them,we assume that
this distance is constant and hence can be encoded into $h_{\alpha\alpha'}$
(recall also that $\A_1 \cap \A_2 = \emptyset$). We also define
the transpositions $F_1 = 2, F_2 = 1$.

Throughout the analysis of incommensurate bilayers, we employ the following standing assumption.

\begin{assumption}
 \label{assump:decay}
   The inter-layer interaction, $h_{\alpha\alpha'}$ for $\alpha \in \A_j, \alpha' \in \A_{F_j}$, is analytic. Moreover, all interactions decay exponentially. More precisely, for $j = 1, 2,$
   \begin{align*}
         |h_{\alpha\alpha'}(R)| &\leq C e^{-\tilde{\gamma}|R|} \qquad
               \text{for all } \alpha, \alpha' \in \A_j, \quad R \in \R_j,  \qquad \text{and}  \\
   |\partial_{x_1}^{m_1} \partial_{x_2}^{m_2}  h_{\alpha\alpha'}(x)| &\leq C_{m}
      e^{-\tilde{\gamma}_{m}|x|} \qquad \text{for } x \in \mathbb{R}^2,
      \quad m \in (\{0\} \cup \mathbb{N})^2, \quad \alpha \in \A_j, \alpha' \in \A_{F_j}.
   \end{align*}
\end{assumption}

We define the Hamiltonian matrix $H$ by
\begin{equation*}
H_{R\alpha,R'\alpha'} = h_{\alpha\alpha'}(R-R') \qquad \text{ for } (R\alpha, R'\alpha') \in \Omega \times \Omega.
\end{equation*}
We then define the intralayer convolution operator $d \ast \psi$ for $d = {\rm diag}(h)$ as
\begin{equation*}
[d \ast \psi ]_\alpha(R) = \sum_{R'\alpha' \in \Omega_j}h_{\alpha\alpha'}(R - R')\psi_{\alpha'}(R')\qquad\text{for }R\alpha \in \Omega_j.
\end{equation*}

Next,  we define shift-dependent interlayer coupling functions
\begin{equation*}
   h^{bj}_{\alpha\alpha'}(x) = h_{\alpha\alpha'}(x -b\delta_{\alpha \in \A_{F_j}}+b \delta_{\alpha' \in \A_{F_j}})
\end{equation*}
and the interlayer coupling convolution for $R\alpha \in \Omega_i$ as
\begin{equation*}
[h^{bj} \ast \psi]_\alpha(R) = \sum_{R'\alpha' \in \Omega_{F_i}}h^{bj}_{\alpha\alpha'}(R-R')\psi_{\alpha'}(R').
\end{equation*}

We can then define the matrix $H_j(b) : \ell^2(\Omega) \rightarrow \ell^2(\Omega)$ by
\begin{equation*}
H_{j}(b)\psi = d \ast \psi +  h^{bj} \ast \psi.
\end{equation*}
This is the Hamiltonian matrix centered at shift $b \in \Gamma_{F_j}^*$.
Notice that intralayer interaction in this model is shift independent.

For any $g \in \mathbb{P}$, $g(H)$ is well defined, and hence the DoS can be defined
as the limit
\begin{equation*}
   \D[H](g) := \lim_{r \rightarrow \infty}\frac{1}{|B_r(0)|}  \Tr_{B_r(0)} g(H),
\end{equation*}
where $\Tr_{B_r(0)}$ is the trace over $\Omega \cap (B_r(0) \times \A)$; see \cite{cances2016}
for a proof that this limit exists. We then recall from \cite{massatt2017} the
representation formula
\begin{align}
\label{e:dos_incomm}
\D[H](g) &= \nu \sum_{j=1}^2 \sum_{\alpha \in \A_j} \int_{b \in  \Gamma_{F_j}}  \bigl[ g\circ H_j(b) \bigr]_{0\alpha,0\alpha} db, \\
\notag
   & \qquad \nu = (|\A_1|\cdot |\Gamma_2| + |\A_2|\cdot |\Gamma_1|)^{-1},
\end{align}
where ``sampling'' over lattice sites (the trace) is replaced by an integral
over relative shifts between the two layers.

\subsection{Momentum Space Formulation for an Incommensurate Bilayer}
\label{subsec:momentum}
We now transform \eqref{e:dos_incomm} to the momentum space setting.
To that end, we define some additional notation, starting with
the reciprocal lattices with associated unit cells
\begin{align*}
   \K_j &= \{ 2\pi(A_j^{-T}) n \text{ : } n \in \mathbb{Z}^2\},\\
   \Gamma_j^* &= \{ 2\pi (A_j^{-T}) \beta \text{ : } \beta \in [0,1)^2 \},\\
   \Omega_j^* &= \K_{F_j} \times \A_j,\\
   \Omega^* &= \Omega_1^* \cup \Omega_2^*.
\end{align*}
Let $\G_j : \ell^2(\R_j \times \I) \to L^2(\Gamma_j^*) \times \I$ denote the Bloch operator defined analogously to $\G$ in \ref{e:bloch}.
We define the Fourier transform for $\alpha \in \A_j, \alpha' \in \A_{F_j}$, by
\begin{equation*}
\hat{h}_{\alpha\alpha'}(q) = \int h_{\alpha\alpha'}(x)e^{-iq\cdot x}.
\end{equation*}

For $q \in {\mathbb R}^2,$ we define an intralayer operator $\G h^q : \ell^\infty(\Omega^*) \rightarrow \ell^\infty(\Omega^*)$ centered at $q$ by
\begin{equation*}
[\G h^q \hat \psi]_\alpha(R^*) = \sum_{\alpha' \in \A_j}[\G_j h]_{\alpha\alpha'}(q+R^*)\hat \psi_{\alpha'}(R^*)
\quad\text{for } R^*\alpha \in \Omega_j^*,\,\hat\psi \in \ell^\infty(\Omega^*).
\end{equation*}
For $R^*\alpha \in \Omega_j^*,$ the inter-layer coupling  is given by
\begin{equation*}
\begin{split}
[\hat{h}^q \ast \hat \psi]_\alpha(R^*) &= \sum_{\tilde R^*\tilde \alpha \in \Omega_{F_j}^*} \hat{h}^q_{\alpha\tilde\alpha}(R^*+\tilde R^*)\hat\psi_{\tilde \alpha}(\tilde R^*) ,\\
& \text{where }  \hat{h}^q( \xi) =\sqrt{ |\Gamma_1^*|\cdot|\Gamma_2^*|}\hat{h}(q+\xi).
\end{split}
\end{equation*}
We define $\wH(q) : \ell^\infty(\Omega^*) \rightarrow \ell^\infty(\Omega^*)$ as a sum of intralayer and interlayer components:
\begin{equation*}
\wH(q)\hat \psi = \G h^q \hat\psi + \hat{h}^q \ast \hat\psi.
\end{equation*}
The intralayer operator $\G h^q$ is block-diagonal, but in principle the
interlayer operator $\hat{h}^q \ast$ is non-local. However, it is exponentially decaying, and hence we can introduce a finite cut-off approximation to make it a local operation.

In what follows, we use the convention that, if $ \psi$ is defined over $\Omega$, then $\psi = (\psi^1,\psi^2)$, where $\psi^j = \psi|_{\Omega_j}$.
The real space system is described in terms of real space lattices $\R_j$ and real space shifts $b \in \Gamma_j$, while momentum space is described in terms of reciprocal lattices $\R^*_j$ with shifts in the reciprocal lattice unit cells $\Gamma_j^*$. To transform between the two descriptions, we define the operator $G_q^{bj} : \ell^2(\Omega) \rightarrow \ell^\infty(\Omega^*)$ such that, for $R^*\alpha \in \Omega_k^*$, we have
\begin{equation}
\label{e:IncomBloch}
[G_q^{bj}\psi]_\alpha(R^*) := |\Gamma_k^*|^{-1/2} e^{(-1)^{j+k}ib \cdot (q/2+R^*)} [\G_k\psi^k]_{\alpha}(q+R^*) .
\end{equation}
This is a transformation from the real space to coupled Bloch waves. The phase factor corrects for the relative shifts in the transformation, as we see from the following lemma.

\begin{lemma}
\label{lemma:transform}
Under Assumptions \ref{assump:incomm} and \ref{assump:decay}, we have
\begin{equation*}
G_q^{bj}[H_j(b)\psi] = \wH(q)G_q^{bj}\psi \qquad \forall \psi \in \ell^2(\Omega).
\end{equation*}
This identity decomposes into intralayer and interlayer interactions as follows:
\begin{align*}
   G_q^{bj}[d \ast \psi] &= \G h^q G_q^{bj}\psi \qquad \text{and} \\
   G_q^{bj}[h^{bj} \ast \psi] &= \hat h^q \ast G_q^{bj}\psi.
\end{align*}
\end{lemma}
\begin{proof}
See Section \ref{proof:transform}.
\end{proof}

Thus, the eigenproblems in real space and momentum space are, respectively,
given by
\begin{align*}
& d \ast \psi + h^{bj}\ast \psi = \epsilon \psi,  \qquad \text{and} \\
& \G h^q \hat{\psi} + \hat h^q \ast \hat \psi = \epsilon \hat \psi.
\end{align*}
Note in particular that, without interlayer interaction, this reverts to simple Bloch theory on each independent sheet.

Finally, we obtain the following expression for the DoS, in analogy with the real space formulation (\ref{e:dos_incomm}).

\begin{thm}
\label{thm:dos}
   Under Assumptions \ref{assump:incomm} and \ref{assump:decay} and for $g \in \mathbb{P}$,
   the Density of States is equivalently given by
\begin{align}
      \label{eq:dos_ms}
\D[H](g) &= \nu^*\sum_{j=1}^2\sum_{\alpha \in \A_j} \int_{\Gamma_j^*}\bigl[g\circ \widehat{H}(q)\bigr]_{0\alpha,0\alpha}  dq, \\
\notag
   & \qquad  \nu^* = (|\A_1|\cdot |\Gamma_1^*| + |\A_2|\cdot|\Gamma_2^*|)^{-1}.
\end{align}
\end{thm}
\begin{proof}
   See Section~\ref{sec:proof:thm:dos}.
\end{proof}

\begin{remark}
The DoS is a bounded operator with respect to $g$ using the $\|\cdot\|_\infty$ norm. Further, $\wH(q)$ is bounded. Therefore these operators can be extended to continuous functions.
\end{remark}

\end{section}

\begin{section}{Momentum Space Convergence}
\label{sec:approx}
We now turn to the development of an approximation algorithm that exploits
the momentum space formulation \eqref{eq:dos_ms} of the density of states.
We will also discuss advantages of the momentum space algorithm over the
real space algorithm from \cite{massatt2017}.


\subsection{Motivation}
\label{sec:approx:motivation}
To begin, we describe a naive approach to approximating the formula \eqref{eq:dos_ms}.
We define a subset
\begin{equation*}
\Omega^*_r =\{R^*\alpha\in\Omega^* : |R^*|\leq r\}
\end{equation*}
of momentum space and the associated Hamiltonian restrictions $\wH_r(q) = \wH(q)|_{\Omega^*_r}$.
The DoS operator is then approximated by
\begin{equation} \label{eq:approx:dos_naive}
\widehat{\D}_r(g) = \nu^* \sum_{j=1}^2\sum_{\alpha \in \A_j} \int_{\Gamma_j^*} [g \circ \wH_r(q)]_{0\alpha,0\alpha}dq.
\end{equation}
To evaluate the DoS at an energy $E,$ we test with the Gaussian $\phi_{\epsilon \kappa}(E) = \frac{1}{\sqrt{2\pi}\kappa}e^{-\frac{(E-\epsilon)^2}{2\kappa^2}}.$ Following the argument of Theorem 2.5 of \cite{massatt2017} essentially verbatim, we then obtain the error bound 
\begin{equation*}
   |\D[H](\phi_{\epsilon \kappa}) - \widehat{\D}_r(\phi_{\epsilon \kappa})|
   \lesssim \kappa^{-7} e^{-\gamma \kappa r}.
\end{equation*}
This error bound, although exponential in $r$, has an undesirable dependence
on $\kappa$, most crucially in the exponent. In particular, the approximation
\eqref{eq:approx:dos_naive} gives us no advantage over the real space method.

We will now discuss why, for specific incommensurate systems and
values of $\epsilon$, which are of physical interest, we can construct an
approximation with a significantly improved error bound
\begin{equation}
\label{e:convergence_advantage}
|\D[H](\phi_{\epsilon \kappa}) - \widehat{\D}_{j\epsilon r}(\phi_{\epsilon\kappa})| \lesssim \kappa^{-3} e^{-\gamma r}.
\end{equation}
Here $\gamma$ is dependent on the exponential decay rate of $\wH(q)$. We note particularly that the exponent is independent of the variance of the Gaussian, which is typically small.

This approximation is not only dependent on the finite cut-off $r$, but on the energy $\epsilon$. Typically we are interested in a range of energies $V \subset \mathbb{R}$. For the analysis we first consider a single energy value $\epsilon \in \R$, but extend our results to a region $V$ in Section \ref{subsec:computation}. The formulation of the approximation is nearly identical to \eqref{eq:approx:dos_naive}:
\begin{equation*}
\widehat{\D}_{j\epsilon r}(g) = \nu^* \sum_{j=1}^2 \sum_{\alpha \in \A_j} \int_{\Gamma_j^*} [g \circ \wH_{j\epsilon r}(q)]_{0\alpha,0\alpha}dq,
\end{equation*}
where $\widehat{H}_{j\epsilon r}(q)$ now denotes the Hamiltonian $\wH(q)$ projected onto a more carefully chosen degree of freedom space  $\Omega_{j\epsilon r}^*(q)$. In particular, allowing dependence of $\Omega_{j\epsilon r}^*(q)$ on $\epsilon$ will turn out to be important. The new degree of freedom space
$\Omega_{j\epsilon r}^*(q)$ is no longer a circular cut-out; instead, we will
choose a more complex  core region and then include a finite cut-off radius surrounding it.


We conclude this motivation by remarking that,
if $   \Omega_{j\epsilon r' }^*(q) \subset  \Omega_r^* $ for some $r' < r$, then we will also have
\begin{equation*}
|\D[H](\phi_{\epsilon \kappa}) - \widehat{D}_r(\phi_{\epsilon \kappa})| \lesssim \kappa^{-3} e^{-\gamma r'}.
\end{equation*}
However, the approximation $\widehat{D}_{j\epsilon r}$ is more efficient in the sense that it significantly reduces the number of degrees of freedom required to obtain this error bound. When we consider an energy range $\epsilon \in V \subset \mathbb{R}$ as in Section \ref{subsec:computation}, we pick a more practical region $\Omega^*_{Vr} = \cup_j\cup_{\epsilon  \in V} \Omega^*_{j\epsilon r}$ that is dependent on all of $V$. Here $V$ is typically a small energy interval so that $\Omega_{j\epsilon r}^*$ does not change drastically as a function of $\epsilon \in V$.

\subsection{Construction of $\Omega^*_{j\epsilon r}$}
\label{sec:approx:region}
Next we discuss how to build the subsets of $\Omega^*$. We let $2\pi A_j^{-T} = (a_{1j}^*,a_{2j}^*)$. We pick the row vectors $a_{1j}^*$ and $a_{2j}^*$ such that the angle between them is acute.
\begin{assumption} \label{as:acute}
We assume $\theta := 2\pi \|A_1^{-T} - A_2^{-T}\|_2 \ll 1$.
\end{assumption}
We observe that the intralayer energy of $\widehat{H}(q)$ is block diagonal, and the difference $\|\epsilon I - \G_jh(q)\|_2$ varies continuously in $q$. Here $\G_j h_{\alpha\alpha'}(q)$ is the intralayer coupling, which is defined only for $\alpha,\alpha' \in \A_j$.  We let
\begin{equation*}
Q(\epsilon,q) :=\eta\max_{j \in \{1,2\}}\|\epsilon I - \G_j h(q)\|_2^{-1}.
\end{equation*}
Here
\begin{equation*}
\eta = \sup_{\psi \in \ell^2(\Omega^*)}  \frac{1}{\|\psi^{(1)}\|_2\|\psi^{(2)}\|_2} (\psi^{(1)},0)\widehat{H}(q) \begin{pmatrix} 0\\ \psi^{(2)}\end{pmatrix}
\end{equation*}
gives the strength of the interlayer interaction.
As in the real space method, we will take finite approximations of the matrix $\widehat{H}(q)$. However, the choice of the cut-out will not be a simple circular region as in the real space method, but instead will be dependent on $Q$. In particular, blocks of $\widehat{H}(q)$ where $Q(\epsilon,q) < 1$ contribute less to the DoS, especially large ``connected'' blocks of $\widehat{H}(q)$ satisfying $Q(\epsilon,\cdot) < 1$.
On the other hand, regions where $Q(\epsilon,\cdot) > 1$ contribute strongly, so we take cut-out regions around connected blocks of degrees of freedom where $Q(\epsilon,\cdot) > 1$. If $\eta = 0$, all sites $R^*\alpha,R'^*\alpha'$ with $R^* \neq R'^*$ are decoupled.

We need to make sense of connectedness over $\Omega^*$. To this end, we define (See Figure \ref{fig:mu_theta})
\begin{equation*}
   \mu_\theta := \frac{1}{2}(\underline{a}^*+\bar{a}^*)
   \qquad \text{where} \quad
   \underline{a}^* := \max\{ |a_{kj}^*|\} \quad \text{and} \quad
   \bar{a}^* := \max_j\{|a_{1j}^*+a_{2j}^*|\}.
\end{equation*}

\begin{figure}[ht]
\centering
\includegraphics[width=.4\linewidth]{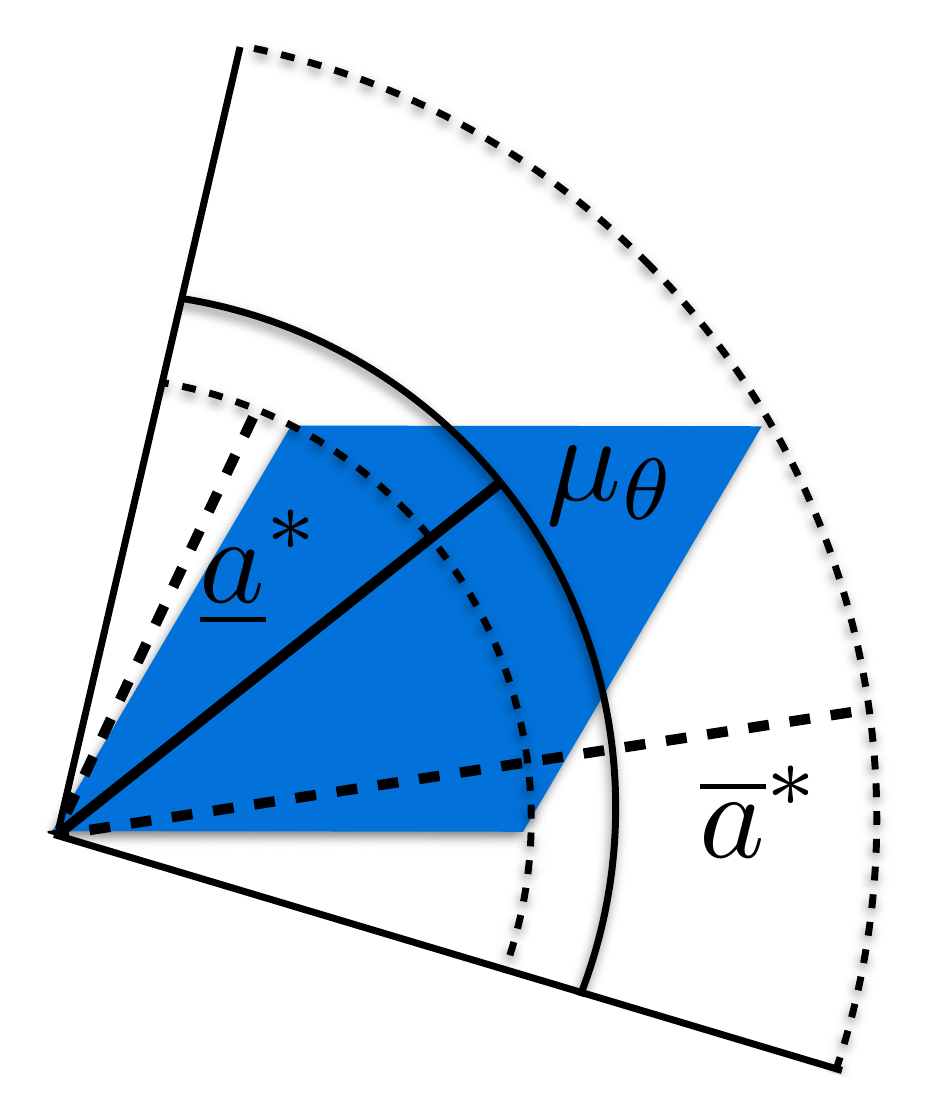}
\caption{ Definition of $\mu_\theta$. }
\label{fig:mu_theta}
\end{figure}

We define a semi-norm on $\Omega^*$
\begin{equation*}
\rho(R^*\alpha,R'^{*}\alpha') = \|R^*-R'^{*}\|_2,
\end{equation*}
and a set of paths from $\omega = R^*\alpha$ to $\omega' = R'^{*}\alpha' \in \Omega^*$ on $U\subset \Omega^*$ by
\begin{equation*}
\mathcal{P}_{\omega\omega'}(U) = \{ \{\omega_j\}_{j=1}^n \subset U \text{ : } \omega_1 = \omega, \omega_n = \omega', n > 0, \text{ and } \rho(\omega_j,\omega_{j+1}) < \mu_\theta\}.
\end{equation*}
\begin{definition}
A set $U \subset \Omega^*$ is called {\em connected} if $\forall \omega,\,\omega' \in U$, $\mathcal{P}_{\omega\omega'}(U) \neq \emptyset$.
\end{definition}

Next, we discuss how to map subsets of a unit cell to $\Omega^*$ and how connectedness on the two regions is related. This description will allow us to determine efficient cut-out and error bounds of the method from the monolayer band structure.

We define $\mathcal{B}(\Omega^*)$ as the set of subsets of $\Omega^*$ and $\mathcal{B}(\Gamma_j^*)$ as the set of Borel sets over $\mathbb{T}(\Gamma_j^*)$. Next we define a mapping $ \lambda_j : \mathcal{B}(\Gamma_j^*)  \rightarrow \mathcal{B}(\Omega^*)$ by
\begin{equation*}
\lambda_j(O) = \{ R^*\alpha \in \Omega^* \text{ : } \nmod_k(R^*) \in A_k^{-T}A_j^{T}O \text{ if } R^*\alpha \in \Omega_k^*,\text{ } k \in \{1,2\}\},
\end{equation*}
where we note that $A_k^{-T}A_j^T : \Gamma_j^*\rightarrow\Gamma_k^*$ is a natural transformation between the two domains and where $\nmod_k(y) = y+R^* \in \Gamma_k^*$ for appropriate $R^* \in \R^*_k$. The maps $\lambda_j$ preserve disjointness. That is, if $O_1 \cap O_2 = \emptyset,$ then $\lambda_j(O_1\cup O_2) = \lambda_j(O_1)\cup\lambda_j(O_2)$ and $\lambda_j(O_1) \cap \lambda_j(O_2) = \emptyset$. Further, if $U \subset \Omega^*$ is connected, then $\nmod_j(\lambda_j^{-1}(U)+B_\theta(0))$ is connected on $\mathbb{T}(\Gamma_j^*)$. This means $\lambda_j$ maps connected sets to a collection of corresponding connected sets in $\Omega^*$.
Hence, $\lambda_j$ is a natural way to map $\mathbb{T}(\Gamma_j^*)$ to $\Omega^*$.

Recall that regions where $Q(\epsilon,\cdot) < 1$ contribute weakly to the DoS while regions where $Q(\epsilon,\cdot) > 1$ contribute strongly. Therefore, we wish to include these regions in our matrix. This motivates the definition
\begin{equation*}
O_{\epsilon,r} = \{q \in \Gamma_j^*,\, j \in \{1,2\} \text{ : } \exists q_0 \in \Gamma_j^* \text{ such that }Q(\epsilon,q_0) > \beta \text{ and } |q-q_0| \leq r \theta\},
\end{equation*}
where $0 < \beta <1$ is chosen such that $\theta \ll \beta^{-1}-1$. Here $\theta$ is the parameter from Assumption~\ref{as:acute}.   For notational simplicity, we don't index $O_{\epsilon,r}$ with $\beta$.

Note in particular that
\begin{equation*}
O_{\epsilon,0} = \cup_{E \in \{\epsilon - \eta\beta^{-1}, \epsilon+\eta\beta^{-1}\}} p(E)
\end{equation*}
where
\begin{equation*}
p(E) = \{ q \in \Gamma_j^*, j \in \{1,2\} :  E = \epsilon_{qn}^j \text{ for some }n \in \{1,\cdots, |\A_j|\} \text{ for some } j\}.
\end{equation*}
Here $\epsilon_{qn}^j$ are the eigenvalues of $\G_jh(q)$. Hence $(O_{\epsilon,0},p,[\epsilon - \eta\beta^{-1}, \epsilon+\eta\beta^{-1}])$ forms a bundle of the level sets of the monolayer band structures. As will be discussed in Remark \ref{remark:criteria}, the homotopy of the total space of the bundle as seen living in one of the tori $\Gamma_j^*$ will determine when this method gains over the real space method.

In Figures \ref{fig:set_plot} and \ref{fig:set_plot2}, we visualize $O_{\epsilon,0}$ and $\lambda_1(O_{\epsilon,0})$ for a graphene bilayer for two different values of $\epsilon$. We can see for $\epsilon = 1.5$ that $\lambda_1$ maps $O_{\epsilon,0}$ to isolated finite regions in $\Omega^*$ while for $\epsilon=2$ $\lambda_1(O_{\epsilon,0})$ is connected on $\Omega^*$.
\begin{figure}[ht]
\centering
\begin{subfigure}{.45\textwidth}
\centering
\includegraphics[width=1\linewidth]{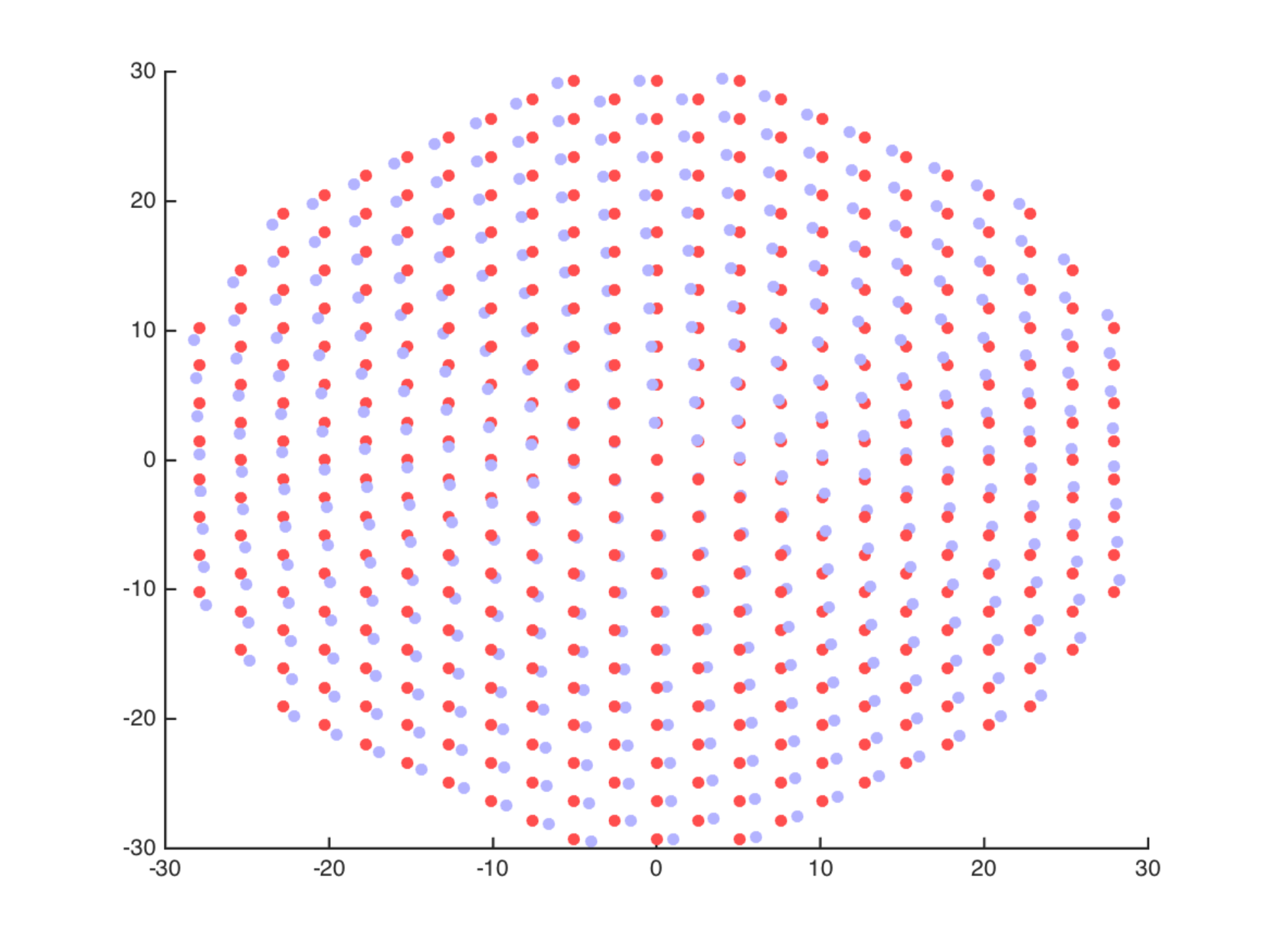}
\caption{We plot $R^* \in \R_1^*\cup\R_2^*$ satisfying $R^*\alpha \in \Omega_{1\epsilon r}^*(q)$ for any $\alpha$ and some $r>0$. Here $q$ is centered at the Dirac point.}
\vspace{13mm}
\end{subfigure}
\hspace{2mm}
\begin{subfigure}{.45\textwidth}
\centering
\includegraphics[width=1\linewidth]{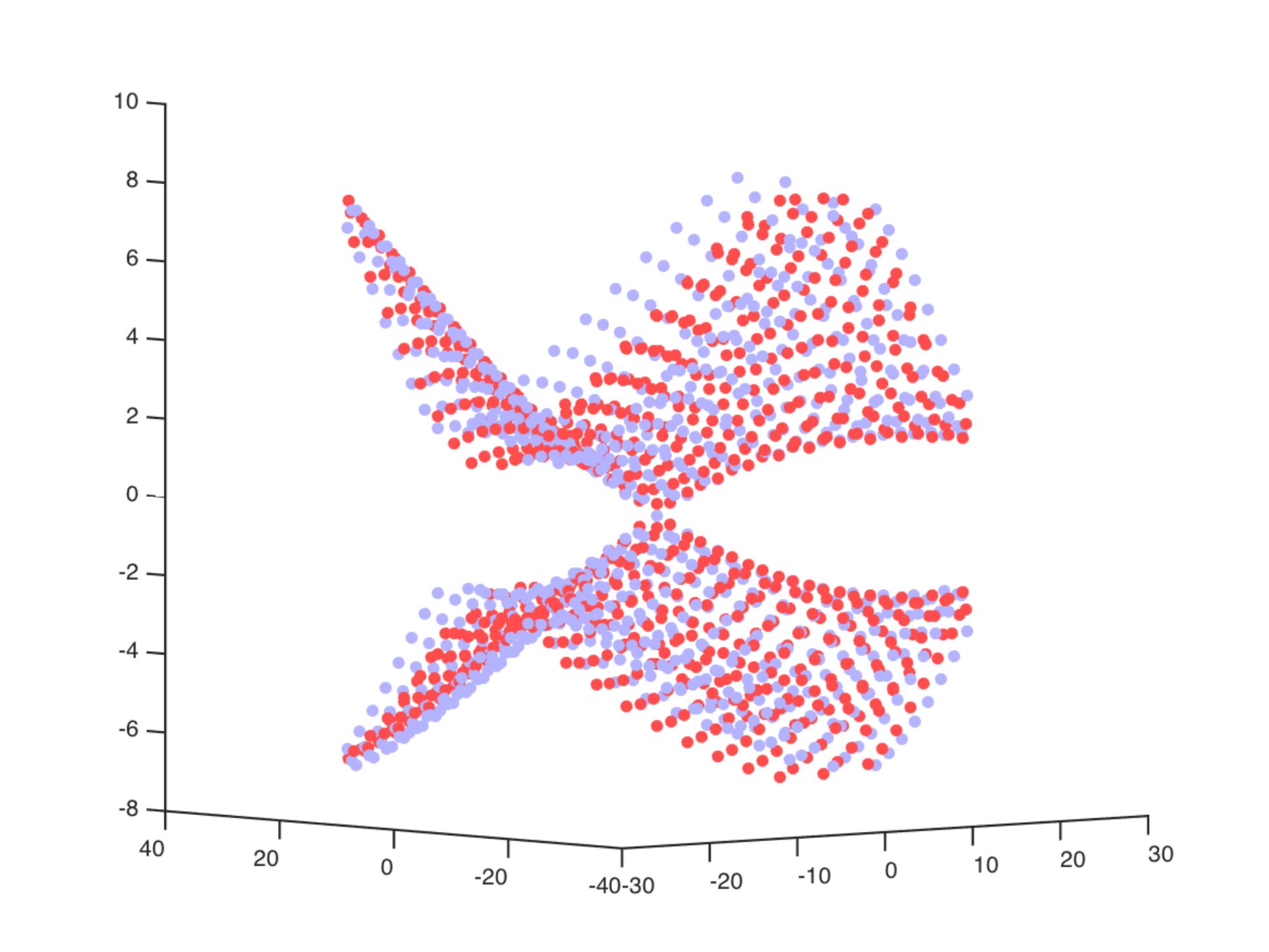}
\caption{ Each degree of freedom $R^*\alpha \in \Omega^*_{1\epsilon r}(q)$ has corresponding intralayer block $\G h^q(R^*)$, where this block is an $|\A_j| \times |\A_j|$ matrix if $\alpha \in \A_j$. Here we plot the eigenvalues of each block against the lattice position $R^*$ for $R^*$ as in Figure A.}
\end{subfigure}
\caption{Visualisation of a local matrix for graphene bilayer with a twist angle of $2^\circ$.}
\label{fig:matrix_rep}
\end{figure}
\begin{figure}[ht]
\centering
\begin{subfigure}{.4\textwidth}
\includegraphics[width=.8\linewidth]{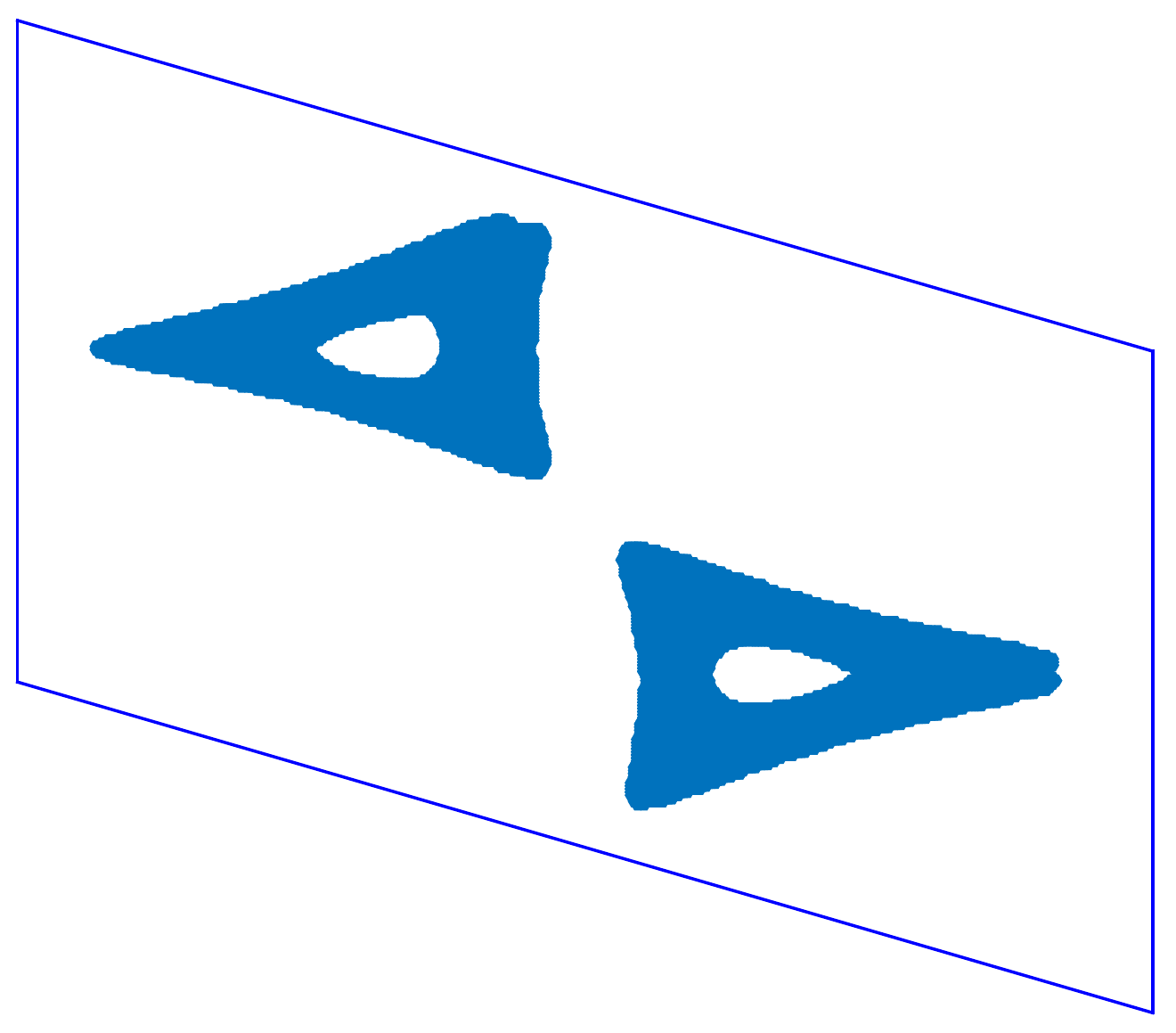}
\caption{We plot $O_{\epsilon,0}\subset \Gamma_1^*$.}
\end{subfigure}
\begin{subfigure}{.4\textwidth}
\vspace{4mm}
\includegraphics[width=.8\linewidth]{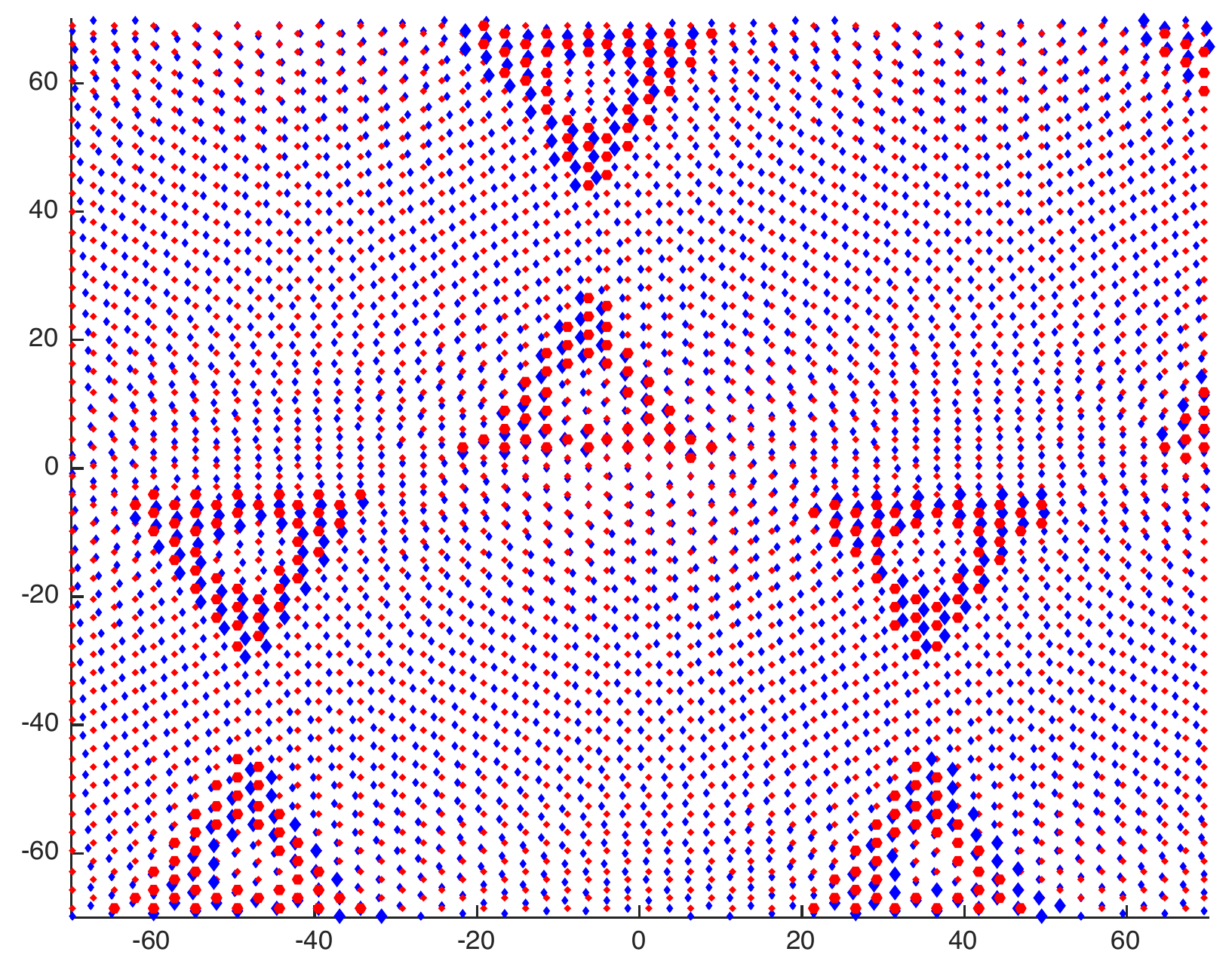}
\caption{We plot $\lambda_1(O_{\epsilon,0}) \subset \Omega^*$.}
\end{subfigure}
\caption{Bilayer graphene with a $2^\circ$ twist with $\epsilon=1.5$.
}
\label{fig:set_plot}
\end{figure}
\begin{figure}[ht]
\centering
\begin{subfigure}{.4\textwidth}
\includegraphics[width=.8\linewidth]{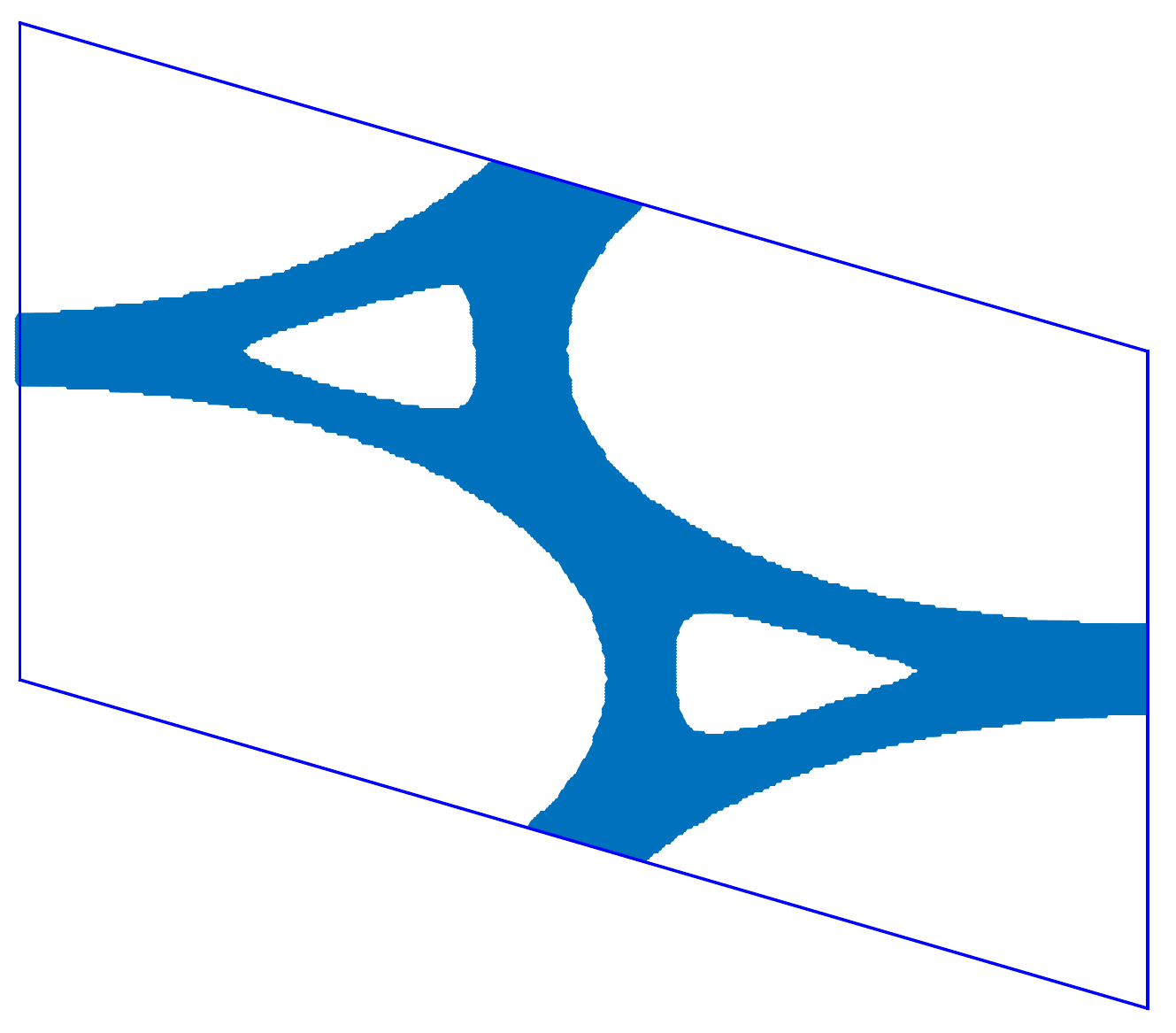}
\caption{We plot $O_{\epsilon,0}\subset \Gamma_1^*$.}
\end{subfigure}
\begin{subfigure}{.4\textwidth}
\vspace{4mm}
\includegraphics[width=.8\linewidth]{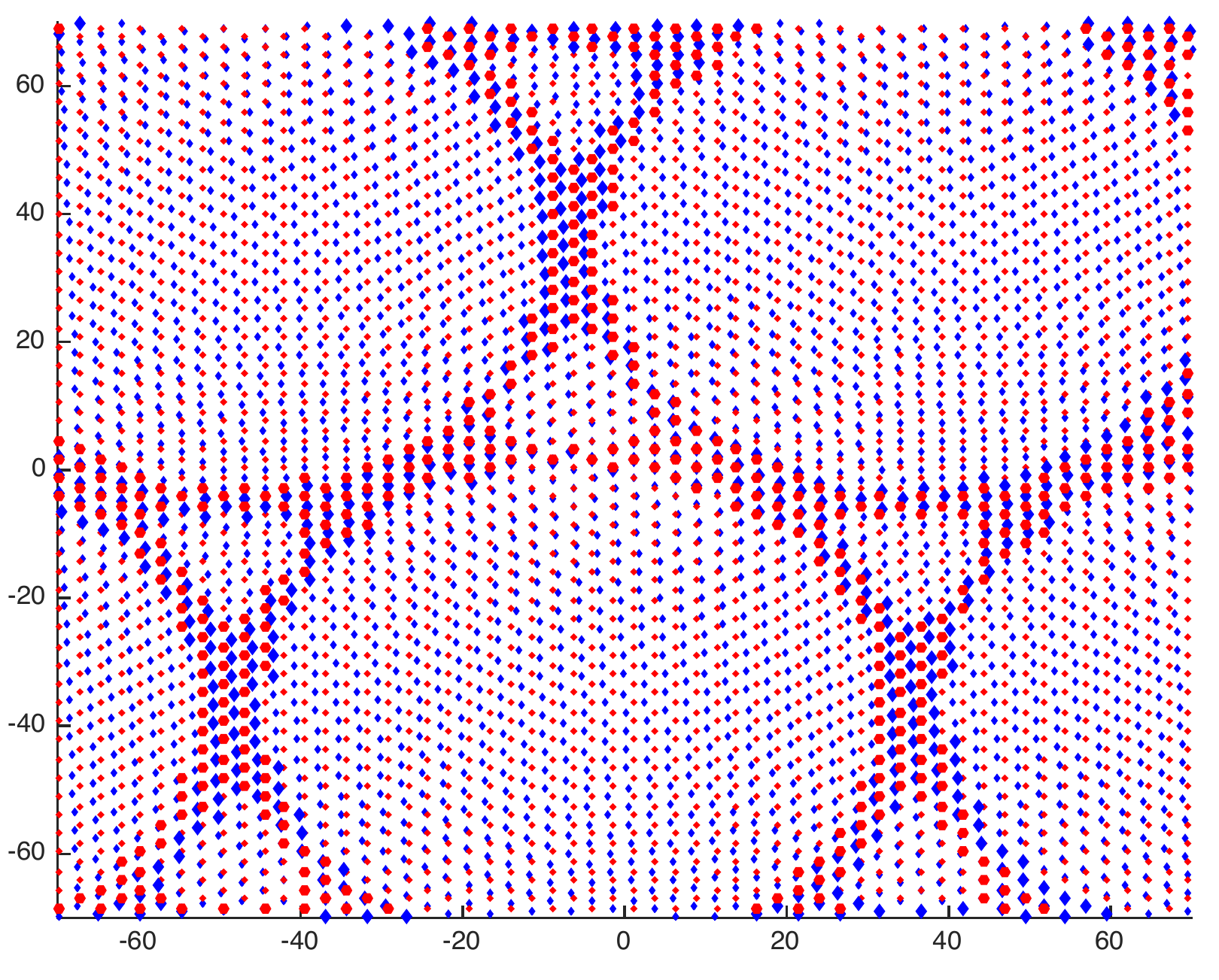}
\caption{We plot $\lambda_1(O_{\epsilon,0}) \subset \Omega^*$.}
\end{subfigure}
\caption{Bilayer graphene with a $2^\circ$ twist with $\epsilon=2$.}
\label{fig:set_plot2}
\end{figure}
\begin{figure}[ht]
\centering
\begin{subfigure}{.4\textwidth}
\includegraphics[width=.8\linewidth]{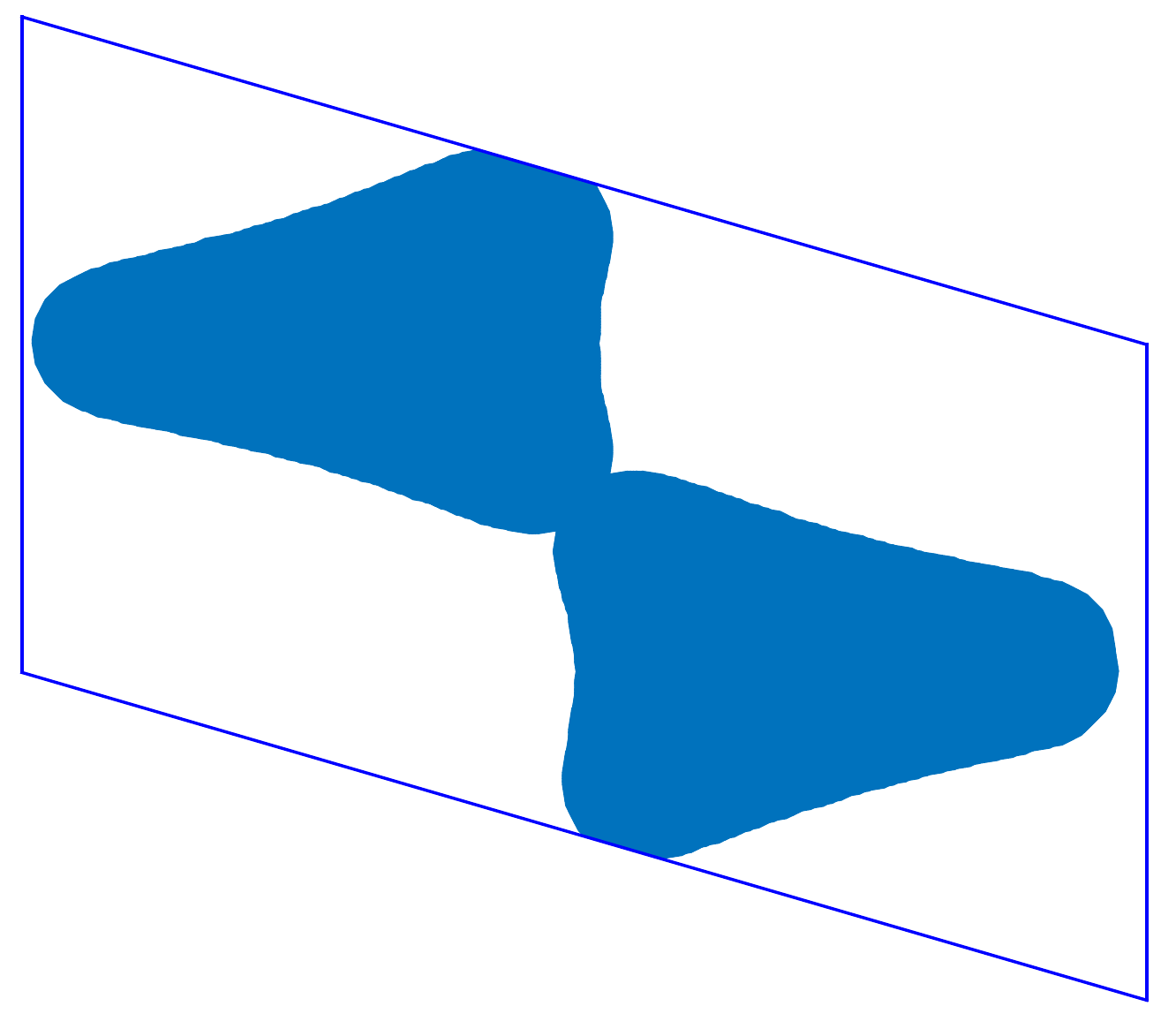}
\caption{We plot $O_{\epsilon,r}\subset \Gamma_1^*$.}
\end{subfigure}
\begin{subfigure}{.4\textwidth}
\vspace{4mm}
\includegraphics[width=.8\linewidth]{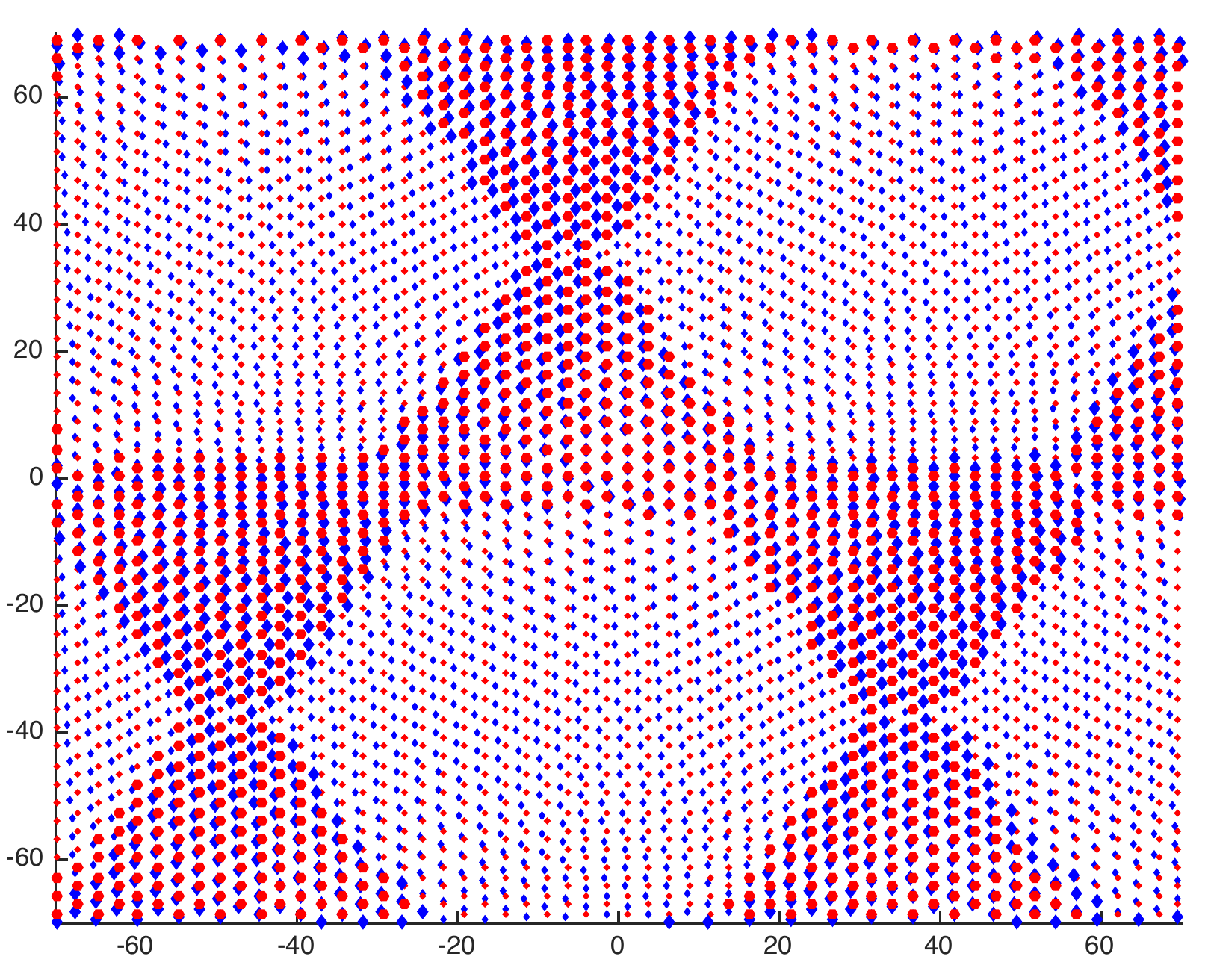}
\caption{We plot $\lambda_1(O_{\epsilon,r}) \subset \Omega^*$.}
\end{subfigure}
\caption{Bilayer graphene with a $2^\circ$ twist with $\epsilon=1.5$ and $r > 0$.}
\label{fig:set_plot3}
\end{figure}

We next define for $U \subset \Omega^*$
\begin{equation*}
\mathcal{P}[U] = \{ \omega \in U \text{ : } \mathcal{P}_{\omega,0\alpha}(U) \neq \emptyset \text{ for any }\alpha  \in \A_j\}.
\end{equation*}
Then we can finally define the degree of freedom space and associated sub-hamiltonian for our approximation scheme as
\begin{align*}
   \Omega^*_{j\epsilon r}(q) &:= \mathcal{P}[\lambda_j(q+O_{\epsilon,r})]
   \qquad \text{and} \\
   \widehat{H}_{j\epsilon r}(q) &:= \widehat{H}(q)|_{\Omega^*_{j\epsilon r}(q)}
\end{align*}
Note that $\Omega^*_{j\epsilon r}(q)$ can be empty, in which case we interpret $[g \circ \widehat{H}_{j\epsilon r}(q)]_{0\alpha,0\alpha} \equiv 0$.

\begin{remark}
\label{remark:criteria}
Suppose for all $q \in O_{\epsilon,r}$ there is a $q' \in O_{\epsilon,r}$ such that $q \in B_\theta(q') \subset O_{\epsilon, r}$. In this case, if $O_{\epsilon,r}$ has non-trivial homotopy group as seen living in one of the tori $\Gamma_j^*$,  $\widehat{H}_{j\epsilon r}(q)$ is an infinite matrix for $q \in O_{\epsilon,r}$. As an example, see Figures~\ref{fig:set_plot2} and~\ref{fig:set_plot3}.

In the first case, we have an infinite matrix, so the approximate matrix is not numerically tractable. In such cases we can use the methodology in \cite{massatt2017} to solve the momentum system with circular cut-out regions, though we gain none of the momentum space convergence advantages of \ref{e:convergence_advantage}. In Figure~\ref{fig:set_plot3}, we reached an infinite system because we took the cut-off radius $r$ to be too large. This simply shows we have a maximum $r>0$ we can choose while keeping \eqref{e:convergence_advantage} numerically tractable. However, in the case where $O_{\epsilon,0}$ has non-trivial homotopy we cannot take advantage of the strong error analysis, and this method loses its advantage. Hence when the total space of the level set bundle $(O_{\epsilon,0},p,[\epsilon-\eta\beta^{-1},\epsilon+\eta\beta^{-1}])$ has non-trivial homotopy, the method does not gain an advantage over the real space method.
\end{remark}

The resulting approximation scheme for the DOS is
\begin{equation*}
\widehat{\D}_{\epsilon r}[H](g) = \nu^*\sum_{j=1}^2\sum_{\alpha \in \A_j} \int_{\Gamma_j^*}\bigl[g\circ \widehat{H}_{j\epsilon r}(q)\bigr]_{0\alpha,0\alpha}  dq.
\end{equation*}
Recall the Gaussian $\phi_{\epsilon\kappa}(x) = \frac{1}{\sqrt{2\pi}\kappa}e^{-\frac{(x-\epsilon)^2}{2\kappa^2}}$, then we have the following theorem:
\begin{thm}
\label{thm:approx}
Given Assumptions \ref{assump:incomm}, \ref{assump:decay}, and \ref{as:acute}, {there exists $0 < \beta_0 < 1$ and $\kappa_0 > 0$ such that if $\beta < \beta_0$ and $\kappa < \kappa_0$, there exists $\gamma > 0$ dependent on $\beta$ and $\hat{h}$ such that}
   \begin{equation*}
   \bigl| \D[H](\phi_{\epsilon\kappa}) - \widehat{\D}_{\epsilon r}[H](\phi_{\epsilon\kappa}) \bigr| \lesssim \kappa^{-3} e^{-\gamma r}.
      \end{equation*}
\end{thm}
\begin{remark}
In the real space method, we also had exponential convergence in $r$, but the convergence rate there was $\gamma \sim \kappa$, while here the convergence rate is independent of $\kappa$. This makes the convergence far faster. In particular, we note that for fixed $\kappa$, we have optimal cut-off radius choice $r \sim \log(\kappa^{-1})$, while in the real space method we had $r \sim \kappa^{-1}\log(\kappa^{-1})$. This allows us to use far smaller matrices (assuming $\wH_{j\epsilon r}(q)$ is finite). In practice, these matrices can easily be small enough to allow us to use full eigensolves instead of the Kernel Polynomial Method. This implies this method will likely be very useful for more complicated electronic objects such as conductivity, where the Kernel Polynomial Method is cumbersome. This will be explored in future work.
\end{remark}

\subsection{Computational Method}
\label{subsec:computation}
In practical computations, we are interested in an energy window $\J \subset \mathbb{R}$. It is therefore preferrable that the sub-hamiltonians we construct are $\epsilon$ independent, though we keep dependence on $\J$. Therefore we let $\Omega_{j\J r}^*(q) = \cup_{\epsilon \in \J}\Omega_{j\epsilon r}^*(q)$ and define the associated $\widehat{H}_{\J r}(q) = \widehat{H}(q)|_{\cup_{j=1}^2\Omega_{j\J r}^*(q)}$. Note that we have also removed the $j$-dependence by slightly increasing the degree of freedom space. This gives the approximation
\begin{equation*}
\widehat{\D}_{\J r}[H](\phi_{\epsilon\kappa}) = \nu^*\sum_{j=1}^2\sum_{\alpha \in \A_j} \int_{\Gamma_j^*}\bigl[\phi_{\epsilon\kappa}\circ \widehat{H}_{\J r}(q)\bigr]_{0\alpha,0\alpha}  dq, \hspace{3mm} \epsilon \in V.
\end{equation*}
It maintains the same error bound as in Theorem \ref{thm:approx}. Here $\J$ is typically a narrow range of energies, and thus the degree of freedom space $\Omega_{j\J r}^*$ does not become too large. For example in the case of bilayer graphene one is typically interested in a short energy interval around the dirac point.

Finally we address quadrature, but for the sake of brevity only give a brief formal discussion. We notice that the region structure $\wH_{\J r}(q)$ is the same for many $q$-points, but centered differently. More specifically, if $R^*\alpha \in \Omega_j^* \cap \Omega^*_{\J r}(q)$, then we have
\begin{equation*}
[\phi_{\epsilon\kappa} \circ \wH_{\J r}(q+R^*)]_{0\alpha,0\alpha} = [\phi_{\epsilon\kappa} \circ \wH_{\J r}(q)]_{R^*\alpha,R^*\alpha}.
\end{equation*}
for $\epsilon \in \J$. If $O_{\epsilon,r}$ separates into $n$ segments on $\Gamma_1^*\cup \Gamma_2^*$, then we let $q_k$, $k \in \{1,\cdots, n\}$ be $q$-points in each region away from the edges of $\Gamma_1^*$ and $\Gamma_2^*$.
We assume $q_k$ is not close to the unit cell edge. Let
\begin{equation*}
\Lambda^* = \{ 2\pi (A_{1}^{-T}-A_2^{-T})\beta \text{ : } \beta \in [0,1)^2\},
\end{equation*}
then
\begin{equation*}
\nu^*\sum_{\alpha \in \A_j}\int_{\Gamma_j^*} [\phi_{\epsilon\kappa} \circ \wH_{\J r}(q)]_{0\alpha,0\alpha}dq =\nu^* \sum_{k=1}^n \int_{\Lambda^*} \Tr [\phi_{\epsilon\kappa} \circ \wH_{\J r}(q_k+q)]dq.
\end{equation*}
We thus have the approximation
\begin{equation}
\label{e:approximation}
\widehat{\D}_{\J  r}[H](\phi_{\epsilon\kappa}) = \nu^*\sum_{k=1}^n \int_{\Lambda^*}\Tr[\phi_{\epsilon\kappa}\circ \widehat{H}_{\J r}(q_k+q)\bigr]  dq,
\end{equation}
and we can now employ a standard uniform discretization of $\Lambda^*$ to evaluate the integral.

With \eqref{e:approximation} we now have a single matrix for computing multiple discretization points simultaneously. When the matrices are small, this also lends quickly to using a full eigensolver. Note that the eigenvectors are unnecessary for the calculations.

\end{section}

\begin{section}{Numerical Tests}
\label{sec:numerics}

\begin{figure}[ht]
\centering
\includegraphics[width=.85\linewidth]{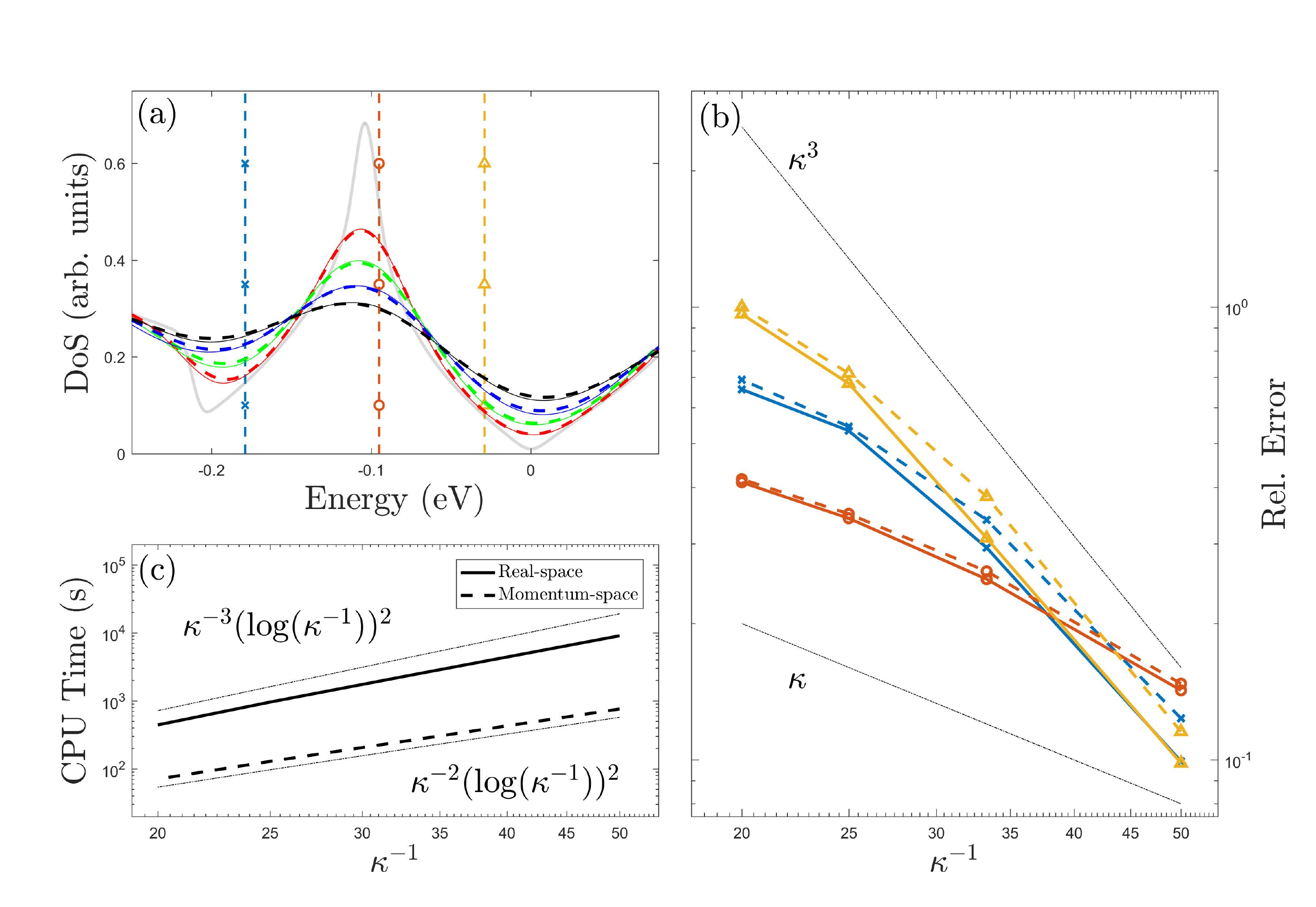}
\caption{(a) Density of states near the Fermi-level from the test calculations. In dark color (Red-Green-Blue-Black) are runs of different accuracy, indexed by the $\kappa$ introduced in the text. Momentum space methods are plotted in dashed lines and real space methods in solid lines. A ``true'' DoS is shown in light-grey, computed using a real space method with parameters chosen for much higher accuracy. Three energy values are marked by the x, circle, and triangle dashed-lines. (b) Relative error convergence towards the true DoS shown for the three marked energies. Also plotted for convenience are grey lines corresponding to $\kappa$ and $\kappa^3$ convergence. (c) Total CPU time for both methods as a function of the accuracy $\kappa$ and two grey lines with the estimated cost scaling for both methods for comparison.}
   \label{fig:numeric_tests}
\end{figure}

To test the accuracy and speed of the approach for calculating the Density of States in \eqref{e:approximation} in practise, the low-energy density of states of $3^\circ$ twisted bilayer graphene was chosen. We use the tight-binding parameters for the bilayer graphene system found in \cite{shiang2016}. The spectrum of this system is well understood, in summary it looks similar to monolayer graphene's spectrum ($\widehat{\D}(\epsilon) = v_f |\epsilon|$) but with additional singular features (Van Hove singularities) within a narrow energy window. A set of calculations of varying accuracy were run for both the real space and momentum space method. Both methods had three parameters to set in order to control accuracy and computational complexity: matrix cut-off radius ($r$), integration sampling number ($N$, giving $N^2$ total sampled points), and either KPM polynomial order ($p$) or Gaussian smoothing width ($\kappa$).

For the momentum space method (indexed by $\kappa$), we keep the cut-off radius fixed and center the lattice at one of the Dirac cones. For testing only the low-energy spectrum, a large matrix is not needed due to the steepness of the Dirac cone and the resulting fast decay of states far from the Dirac cone's center (it is also why the ``k-dot-p'' approach has been studied extensively for twisted bilayer-graphene in the physics literature \cite{bistritzer2011, kormanyos2015}). Since this matrix is very small ($76 \times 76$) a direct eigensolve is performed and each eigenvalue is smoothed by a Gaussian of width $\kappa$. The only significant scaling in computational complexity comes from growth in the number of integration points for evaluating the trace, as a very thin Gaussian width requires many eigenvalues to properly resolve the density of states.

For the real space method (indexed by $p$), the number of sampling points is kept fixed at relatively small value as the DoS is extremely smooth with respect to real space shift \cite{carr2017}. Instead the cut-off radius is increased as one increases $p$. This approach uses the Kernel Polynomial Method (KPM) and thus the Hamiltonian needed to be rescaled to ensure its entire spectrum lay in the interval [-1,1] by dividing the entire matrix by $E_b = 13 \textrm{ eV}$. The rescaled value then naturally enters into the relation between $\kappa$ and $p$. A summary of the relevant six parameters is as follows:

\begin{align*}
\kappa &= (m+1)/100, m \in [1,4] \\
r_\kappa &= 7  \\
N_\kappa &= (2/3)\kappa^{-1} \log(\kappa^{-1}) \\
p &= E_b \pi \kappa^{-1} \\
r_p &= (3/200) p \log(p) \\
N_p &= 10. \\
\end{align*}

To compute convergence rates for both methods, a true-value needed to be specified, so a real space calculation with $p = 8172$ and $r_p = 1105$ \AA$ $ fills this role. The results for these numerical tests are shown (Figure \ref{fig:numeric_tests}). Excellent agreement occurs between the real space (solid lines) and momentum space methods (dashed lines) for the four different values of $\kappa$. The convergence rate is plotted for three different energy values, and it varies between $\kappa^{-1}$ and $\kappa^{-3}$ depending on whether the derivative(s) of the density-of-states operator are zero at that energy value. We see the convergence rate is very similar for both the real space and momentum space methods, and two of the sample points are past the $\kappa_0$ that yields the $\kappa^{-3}$ convergence of Theorem \ref{thm:approx}.

The momentum space method's computational complexity scales like
\begin{equation*}
N_\kappa^2 \sim \kappa^{-2} (\log(\kappa^{-1}))^2,
\end{equation*}
as the matrix-size is kept fixed and changing $\kappa$ does not change the cost of Gaussian smoothing. The real space method scales like the size of the matrix, $r_\kappa^2$, as well as the polynomial order, $\kappa^{-1}$, giving total cost of
\begin{equation*}
r_\kappa^2 \kappa^{-1} \sim \kappa^{-3} (\log(\kappa^{-1}))^2.
\end{equation*}
 The momentum space method is not only faster but also has better asymptotic scaling when compared to a real space calculation of the same accuracy. These results validate the momentum space method and show the significant speed-up it provides over the real space approach in the bilayer graphene system.

\end{section}

\begin{section}{Conclusions}
We have derived a corresponding momentum space formulation for the real space incommensurate system. This results in an alternative numerical scheme where the convergence rate becomes strongly dependent on the moir\'e pattern and the monolayer electronic structure, in particular the band structure. It is shown that the homotopy groups of the band structure level sets determine the efficiency of this algorithm.

In particular, this method has no advantage over the real space method when the band structure level set bundle with width equal to the interlayer coupling energy admits non-trivial homotopy (Remark \ref{remark:criteria}).
However, for certain materials and energy ranges, this method converges asymptotically faster than the real space method, and promises to be very efficient for more complex electronic observables such as conductivity.

In a future work we will also use the method introduced here as an analytical tool to study the validity of the popular supercell approximation~\cite{Terrones2014,Loh2015,Ebnonnasir2014,Koda2016,Komsa2013} and prove sharper error bounds for the real space method \cite{massatt2017}.
\end{section}

\begin{section}{Proofs}
\label{sec:proofs}
\subsection{Proof of Lemma \ref{lemma:transform}}
\label{proof:transform}
Before we can prove the final result, we need to introduce additional notation for transforming between real space to momentum space, and prove a few properties of these transformations.
We define $G_q : \ell^2(\Omega) \rightarrow \ell^\infty( \Omega^*)$ by
\begin{equation*}
[G_q\psi]_{\alpha}(R^*) =
|\Gamma_j^*|^{-1/2}[\G_j\psi^k]_\alpha(q+R^*) \text{ for }R^*\alpha \in \Omega_k^*.
\end{equation*}
Recall the definition $G_q^{bj} : \ell^2(\Omega) \rightarrow \ell^\infty( \Omega^*)$ {from (\ref{e:IncomBloch}}),
\begin{equation*}
[G_q^{bj}\psi]_{\alpha}(R^*) =
|\Gamma_j^*|^{-1/2}e^{(-1)^{j+k}ib\cdot(q/2+R^*)}[\G_j\psi^k]_\alpha(q+R^*) \text{ for }R^*\alpha \in \Omega_k^*.
\end{equation*}
This transforms from real space to momentum space.
We define the projection $P_k : \ell^\infty(\Omega^*) \rightarrow \ell^\infty(\Omega^*)$ to
project onto vector components of sheet $k$. Specifically, for $\hat \psi \in
\ell^\infty(\Omega^*)$,
\begin{equation*}
   P_k\hat\psi = (\delta_{1k}\hat \psi^1,\delta_{2k}\hat\psi^2).
\end{equation*}
We define $\tilde P_j : \ell^2(\Omega) \rightarrow \ell^2(\Omega^j)$ such that $\tilde P_j\psi = \psi^j$ for $\psi \in \ell^2(\Omega)$.
Let $\tilde b = (-1)^{1+j}b$, then we apply $\G_1$ to obtain
\begin{equation*}
\G_1\circ \tilde P_1[H_j(b) \psi](q) = [\G_1h^{11}](q)[\G_1 \psi^1](q)
   + \Big\{ \G_1\Big[\sum_{R_2 \in \R_2,\alpha' \in \A_2} h_{\alpha_1\alpha_2}(R_1-R_2+\tilde b) \psi_{\alpha_2}(R_2)\Big](q)\Big\}_{\alpha_1 \in \A_1}.
\end{equation*}
Here we define $\G_1 h^{11}(q) = \G_1 h_{\alpha\alpha'}(q) |_{\alpha,\alpha' \in \A_1}.$
Next we substantially rewrite the second term on the right-hand side.
Recall that the Fourier transform $\hat{h}_{\alpha\alpha'}$ for
$\alpha,\alpha'$ in different sheets satisfies
\begin{equation*}
h_{\alpha\alpha'}(x) = \int_{\mathbb{R}^2} \hat{h}_{\alpha\alpha'} (\xi)e^{-i\xi\cdot x}d\xi.
\end{equation*}
We then have, for each $q \in \mathbb{R}^2$ and $\alpha_1 \in \A_1$,
\begin{equation*}
\begin{split}
\hspace{2cm} & \hspace{-2cm} \G_1\Big[ \sum_{R_2 \in \R_2,\alpha_2 \in \A_2} h_{\alpha_1\alpha_2}(R_1-R_2+\tilde b) \psi_{\alpha_2}(R_2)\Big](q) \\
&= \sum_{R_2 \in \R_2, R_1 \in \R_1, \alpha_2 \in \A_2} h_{\alpha_1\alpha_2}(R_1-R_2+\tilde b)  \psi_{\alpha_2}(R_2) e^{-i q \cdot R_1} \\
&=\sum_{R_2 \in \R_2, R_1 \in \R_1, \alpha_2 \in \A_2} \int_{\mathbb{R}^2} \hat{h}_{\alpha_1\alpha_2}(\xi) e^{-i\xi \cdot (R_1-R_2+\tilde b)} d\xi \, \psi_{\alpha_2}(R_2) e^{-i q \cdot R_1}\\
&=\sum_{R_1 \in \R_1, \alpha_2 \in \A_2}
      \int_{\mathbb{R}^2} \hat{h}_{\alpha_1\alpha_2}(\xi) e^{-i\xi \cdot (R_1+\tilde b)}
      \Big(\sum_{R_2 \in \R_2} e^{i \xi \cdot R_2} \psi_{\alpha_2}(R_2) \Big)
      d\xi \,  e^{-i q \cdot R_1} \\
&= \sum_{\alpha_2 \in \A_2} \int_{\mathbb{R}^2} \hat{h}_{\alpha_1\alpha_2}(\xi) \sum_{R_1 \in \R_1}e^{i (-\xi + q)\cdot R_1} e^{-i\xi\cdot \tilde b}[ \G_2\psi^2]_{\alpha_2}(\xi) \ d\xi.
\end{split}
\end{equation*}
We note that $\sum_{R_1 \in \R_1} e^{i x \cdot R_1} = |\Gamma_1^*| \sum_{R_1^* \in \K_1}\delta(x-R_1^*)$ in a distributional sense. We obtain
\begin{equation*}
\begin{split}
   & [G_q P_1H_j(b) P_2\psi]^1_{\alpha_1}(R^*_2)  \\
&=|\Gamma_1^*|^{-1/2} \G_1\Big[\sum_{R_2 \in \R_2,\alpha_2 \in \A_2} h_{\alpha_1\alpha_2}(R_1-R_2+\tilde b) \psi_{\alpha_2}(R_2)\Big](q+R^*_2) \\
&= |\Gamma_1^*|^{1/2}\sum_{\alpha_2 \in \A_2, R^*_1 \in \K_1} e^{-i(q+R^*_1+R^*_2)\cdot \tilde b}\hat{h}_{\alpha_1\alpha_2}(q+ R^*_1 +R^*_2)[\G_2\psi^2]_{\alpha_2}(q + R^*_1)\\
&=|\Gamma_1^*|^{1/2}|\Gamma_2^*|^{1/2}\sum_{\alpha_2 \in \A_2, R^*_1 \in \K_1} e^{i(q+R^*_1+R^*_2)\cdot (-1)^j b}\hat{h}_{\alpha_1\alpha_2}(q + R^*_1+R^*_2)[G_q\psi]_{\alpha_2}(R^*_1) \\
&=|\Gamma_1^*|^{1/2}|\Gamma_2^*|^{1/2}\sum_{\alpha_2 \in \A_2, R^*_1 \in \K_1}
e^{i(-1)^jq \cdot b/2} e^{iR^*_2\cdot(-1)^j b}\hat{h}_{\alpha_1\alpha_2}(q + R^*_1+R^*_2)  \\
   & \hspace{6cm} \cdot e^{i R^*_1 \cdot(-1)^jb}e^{i(-1)^jq\cdot b/2}[G_q\psi]_{\alpha_2}(R^*_1).
\end{split}
\end{equation*}
Therefore we can conclude that
\begin{equation}
\label{e:part1}
P_1G_q^{bj}H_j(b)P_2\psi=P_1\widehat{H}(q) G_q^{bj}P_2\psi.
\end{equation}

We treat the intralayer interaction analogously:
\begin{equation*}
\begin{split}
& [G_q H_j(b) P_1\psi]_\alpha^1(R^*_2)  \\
&= |\Gamma_1^*|^{-1/2} \G_1[\sum_{R_1 \in \R_1,\alpha' \in \A_1}h_{\alpha_1\alpha_2}(R-R')\psi_{R_2\alpha_2}](q+R^*_2)\\
&= |\Gamma_1^*|^{-1/2}\sum_{\alpha_1 \in \A_1} [\G_1 h]_{\alpha_1\alpha_2}(q + R^*_2) [ \G_1\psi^1]_{\alpha_2}(q+R^*_2) \\
&= |\Gamma_1^*|^{-1/2}\sum_{\alpha_1 \in \A_1} e^{i (-1)^{j}(R^*_1+ q \cdot b/2)}[\G_1 h]_{\alpha_1\alpha_2}(q + R^*_2)e^{-i (-1)^{j}(R^*_1+ q \cdot b/2)} [ \G_1\psi^1]_{\alpha'}(q+R^*_2),
\end{split}
\end{equation*}
that is,
\begin{equation}
\label{e:part2}
P_1 G_q^{bj} H_j(b)P_1\psi =P_1 \widehat{H}(q)G_q^{bj}{P_1}\psi.
\end{equation}
Combining {\eqref{e:part1} and \eqref{e:part2}}, we obtain that
\begin{equation*}
P_1 G_q^{bj} H_j(b)\psi =P_1 \widehat{H}(q)G_q^{bj}\psi.
\end{equation*}
The same argument holds for the second sheet, and thus
\begin{equation*}
G_q^{bj}(H_j(b)  \psi) = \widehat{H}(q) G_q^{bj}\psi.
\end{equation*}
This completes the proof of Lemma \ref{lemma:transform}.

\subsection{Proof of Theorem \ref{thm:dos}}
 \label{sec:proof:thm:dos}

For all $g \in \mathbb{P}$, we have
\begin{equation*}
G_q^{bj} (g\circ H_j(b)\psi) = g\circ\widehat{H}(q)G_q^{bj}\psi.
\end{equation*}
This follows trivially from Lemma \ref{lemma:transform}.
We define $e^{\alpha} \in \ell^2(\Omega)$,
i.e. $e_{\alpha'}^\alpha(R) = \delta_{R0}\delta_{\alpha\alpha'}$. Additionally we define
$\hat e^\alpha \in \ell^2(\Omega^*)$ by
\begin{equation*}
\hat e_{\alpha'}^\alpha(R^*) = \delta_{R^*0}\delta_{\alpha\alpha'}
\qquad \text{for $R^*\alpha' \in \tilde \Omega$,}
\end{equation*}
and note that for $\alpha \in \A_k$,
\begin{equation*}
G_qe^\alpha = |\Gamma_k^*|^{-1/2}I^\alpha \qquad \text{where} \quad
I^\alpha_{\alpha'}(R^*) = \delta_{\alpha\alpha'}.
\end{equation*}
We define $\Phi^{bj} : \ell^\infty(\Omega^*) \rightarrow \ell^\infty(\Omega^*)$
\begin{equation*}
[\Phi^{bj}\psi]_{\alpha}(R^*) = e^{(-1)^{j+k}ib\cdot R^*}\psi_\alpha(R^*).
\end{equation*}
Then $G_q^{bj} e^\alpha = e^{(-1)^{j+k}b\cdot q/2}\Phi^{bj}I^\alpha$.
Next, we recall the Local Density of States (LDoS) operator { \cite{massatt2017} } for $g \in \mathbb{P}$, is given by
  \begin{equation}
  \label{e:ldos}
 \D_\alpha[H](b,g) := [e^\alpha]^*\bigl( g \circ H_{j}(b)\bigr) e^\alpha, \qquad \alpha \in \A_j.
  \end{equation}

\begin{lemma}
\label{lemma:local_dos}
   Under Assumption \ref{assump:incomm} (incommensurate bilayer) and
   Assumption \ref{assump:decay} (locality of $H$)
    we have, for $\alpha \in \A_{j}$ and for $g \in \mathbb{P}$,
\begin{equation*}
\D_\alpha[H](b,g) =\mint_{\Gamma_j^*}[\hat e^\alpha]^* \bigl( g\circ \widehat{H}(q)\bigr) \Phi^{bj}I^\alpha dq.
\end{equation*}
\end{lemma}
\begin{proof}
{We wish to represent the LDoS from \eqref{e:ldos}} in terms of $\widehat{H}(q)$.
We let $\psi = \bigl(g \circ H_{{j}}(b)\bigr)e^\alpha \in \ell^2(\Omega)$. If $\alpha \in \A_j$, then
\begin{equation*}
\begin{split}
[e^\alpha]^*\psi &= \mint_{\Gamma_j^*}[\G_1\psi](q)dq
= |\Gamma_j^*|^{-1/2}[\hat e^\alpha]^*\int_{\Gamma_j^*}e^{-(-1)^{j+k}b\cdot q/2} G_q^{bj}\psi,
\end{split}
\end{equation*}
which we insert into \eqref{e:ldos} to obtain
\begin{equation*}
\D_\alpha[H](b,g) = [g \circ H_j(b)]_{0\alpha,0\alpha}= |\Gamma_j^*|^{-1/2}[\hat e^\alpha]^* \int_{\Gamma_j^*}e^{-(-1)^{j+k}b\cdot q/2} \Phi^{bj}\bigl(g \circ \widehat{H}(q) \bigr)G_q^{bj} e^\alpha.
\end{equation*}
However, $G_q^{bj} e^\alpha = e^{(-1)^{j+k}b\cdot q/2}\Phi^{bj}I^\alpha$, hence
\begin{equation*}
\D_{\alpha}[H](b,g) = \mint_{\Gamma_j^*} [\hat e^\alpha]^*\bigl(g \circ \widehat{H}(q)\bigr) \Phi^{bj} I^\alpha dq,
\end{equation*}
which is the desired result.
\end{proof}

We can now prove the Theorem \ref{thm:dos}. Using Lemma \ref{lemma:local_dos}
we have
\begin{equation*}
\begin{split}
\D[H](g) & = \nu \sum_{j=1}^2\sum_{\alpha \in \A_j}\int_{\Gamma_{F_j}}\D_\alpha[H](b,g)db \\
&= \nu \sum_{j=1}^2\sum_{\alpha \in \A_j} \mint_{\Gamma_j^*} [\hat e^\alpha]^*\bigl(g \circ \widehat{H}(q)\bigr) \biggr(\int_{\Gamma_{F_j}}\Phi^{bj} I^\alpha db\biggl)dq\\
&=\nu \sum_{j=1}^2\sum_{\alpha \in \A_j} \mint_{\Gamma_j^*} [\hat e^\alpha]^*\bigl(g \circ \widehat{H}(q)\bigr) \hat e^\alpha |\Gamma_{F_j}|dq\\
&=\nu\sum_{j=1}^2|\Gamma_{F_j}|\sum_{\alpha \in \A_j} \mint_{\Gamma_j^*} \bigl[g \circ \widehat{H}(q)\bigr]_{0\alpha,0\alpha} dq\\
&=\nu^*\sum_{j=1}^2|\Gamma_j^*|\sum_{\alpha \in \A_j} \mint_{\Gamma_j^*} \bigl[g \circ \widehat{H}(q)\bigr]_{0\alpha,0\alpha} dq\\
&=\nu^*\sum_{j=1}^2\sum_{\alpha \in \A_j} \int_{\Gamma_j^*} \bigl[g \circ \widehat{H}(q)\bigr]_{0\alpha,0\alpha} dq,\\
\end{split}
\end{equation*}
which completes the proof of Theorem \ref{thm:dos}.

\subsection{Proof of Theorem \ref{thm:approx}}
We need to show that, for $\kappa$ sufficiently small,
   \begin{equation*}
   \bigl| \D[H](\phi_{\epsilon\kappa}) - \widehat{\D}_{\epsilon r}[H](\phi_{\epsilon\kappa}) \bigr| \lesssim \kappa^{-3} e^{-\gamma r}.
      \end{equation*}
We define a sequence $\{r_n\}_{n=0}^N$ such that $r_{N-1} = r$, $r_0 = 0$, $r_N = \infty$, and $r_n-r_{n-1} = \mu_\theta$ for $1 \leq n < N-1$. Furthermore, we define the regions
\begin{equation*}
   \Omega^* \supset \Omega_n^* :=
   \begin{cases}
      \lambda_j(q+O_{\epsilon,r_n}) , & n < N, \\
      (\Omega^* \setminus \Omega_{2N-n-1}^*) \cup \mP[\Omega_{N-1}^*], & N \leq n \leq 2N-1, \\
      \Omega^*, & n = 2N.
   \end{cases}
\end{equation*}
The second and third cases, $\Omega_n^*$ for $n \geq N$, are never used in numerical simulations, but are convenient for the analysis.

Next, we define the mapping $J_{U} : \Omega^* \rightarrow U$  for $U \subset \Omega^*$ such that
\begin{equation*}
(J_{U}\psi)_{\alpha}(R) =
{\begin{cases}
\psi_\alpha(R), & R\alpha \in U,\\
0, & \text{ otherwise}.
\end{cases}}
\end{equation*}
We define the associated projection $P_U : \Omega^* \rightarrow \Omega^*$ for $U \subset \Omega^*$ by
\begin{equation*}
P_U = J_U^*J_U.
\end{equation*}
We let $I_U : U \rightarrow U$ be the identity operator. We then define the following operators:
\begin{align*}
\oI_n &: \mP[\Omega^*_n]\rightarrow  \mP[\Omega^*_n]
      & \oI_n &= I_{\mP[\Omega^*_n]}\\
\oJ_n &: \Omega^* \rightarrow \mP[\Omega^*_n]
      & \oJ_n &= J_{\mP[\Omega^*_n]}\\
\partial\oJ_n &: \Omega^*\rightarrow\mP[ \Omega_n^*]\setminus \mP[\Omega_{n-1}^*]
      & \partial \oJ_n &= J_{\mP[ \Omega_n^*]\setminus \mP[ \Omega_{n-1}^*]}\\
\wH_{nm} &: \mP[ \Omega_m^*]\setminus\mP[\Omega_{m-1}^*] \rightarrow \mP[ \Omega_n^*]\setminus \mP[\Omega_{n-1}^*]
      &  \wH_{nm} &= \partial\oJ_n\widehat{H}(q)\partial\oJ_m^*\\
I_{nn} &: \mP[ \Omega_n^*]\setminus \mP[\Omega_{n-1}^*] \rightarrow \mP[\Omega_n^*]\setminus \mP[\Omega_{n-1}^*]
      &I_{nn} &= \partial\oJ_n I \partial\oJ_n^*\\
 J_{nm} &:  \mP[\Omega^*_m] \rightarrow \mP[\Omega^*_n]
      &J_{nm} &= \oJ_n\oJ_m^*\\
\oP_n &:\Omega^* \rightarrow \Omega^*
      &\oP_n &= P_{\mP[\Omega_n^*]}\\
\wH_n &: \Omega^* \rightarrow \Omega^*
      & \wH_n &= \oP_n\widehat{H}(q)\oP_n^*\\
\wH_n' &: \mP[\Omega^*_n] \rightarrow \mP[\Omega^*_n]
      &\wH_n' &= \oJ_n\widehat{H}(q)\oJ_n^*.
\end{align*}
  We suppress the dependence on $\epsilon$, $j$, and $q$ for notational brevity. It suffices then for us to show the following bound ($ n < N$):
\begin{equation}
\label{e:local_bound}
\biggl\|\oJ_0 [\phi_{\epsilon\kappa} \circ \wH_{2N} - \phi_{\epsilon\kappa}\circ \wH_n] \oJ_0^*\biggr\|_2 \lesssim \kappa^{-1} e^{-\gamma r_n}.
\end{equation}
Note that $\wH_{2N} = \widehat{H}(q)$. To prove \eqref{e:local_bound}, we first prove the following Lemma:

\begin{lemma}
   Let $z \in \mathbb{C}$ such that $|z- \epsilon - i\kappa| \ll \kappa$,
   $\ell < k \leq n\leq N$, $k < N$. {Then there exists $\beta_0 > 0$ such that, if $\beta < \beta_0$, then there exists $\gamma' > 0$ such that}
\begin{equation*}
\| J_{kn}(z \oI_n - \wH_n')^{-1}J_{n\ell}\|_2 \lesssim \kappa^{-1}e^{-\gamma'k}.
\end{equation*}
{Note that $\gamma'$ is dependent on $\beta$, $\eta$, and the decay rate of $\widehat{h}$.}
\end{lemma}
\begin{proof}
We define
\begin{align*}
G_{k\ell} &:= J_{kn}(z \oI_n - \wH_n)^{-1}J_{n\ell} \quad \text{and} \\
X &:= \{1,\cdots, 2N-1\}.
\end{align*}
Then, applying the Schur Complement, we have
\begin{align*}
   G_{k\ell} &=  -(z I_{kk} - \wH_{kk})^{-1} \wH_{k0}G_{0\ell} -(z I_{kk} - \wH_{kk})^{-1} \wH_{k,2N}G_{2N,\ell} \\
   &\qquad -\sum_{s \in X\setminus\{k\}} (z I_{kk} - \wH_{kk})^{-1} \wH_{ks}G_{s\ell}.
\end{align*}
Next, let the operator $V : \ell^2(X) \rightarrow \ell^2(X)$ be defined by
\begin{equation*}
V_{ks}:=\eta^{-1}\beta \|\wH_{ks}\|_2(1-\delta_{ks}).
\end{equation*}
{We recall $\widehat{H}(q)$ exhibits exponential decay since $\hat h^q$ decauys exponentially and $d$ is local.} Hence there exists {$\gamma^*>0$} such that
\begin{equation*}
V_{ks}\leq \beta e^{-\gamma^*(|k-s|-1)}(1-\delta_{ks}).
\end{equation*}
Let $\tbeta := e^{\gamma^*}\beta,$ then we also have, for $k \neq 0$ and $k \neq 2N$,
\begin{equation*}
\|(z I_{kk} - \wH_{kk})^{-1}\|_2 \leq \beta \eta^{-1}.
\end{equation*}
Recall $\beta < 1$ is chosen when defining $Q(\cdot,\cdot)$. We also have that
\begin{equation*}
x_s := \|G_{s\ell}\|_2 \leq \kappa^{-1}.
\end{equation*}
We therefore have
\begin{equation}
\label{e:temp1}
\|G_{k\ell}\|_2 \leq \tbeta e^{- \gamma^* |k|}\kappa^{-1} + \tbeta e^{-\gamma^* |2N-k|} \kappa^{-1}+ \sum_{s \in X\setminus\{k\}} \beta\eta^{-1} \|\wH_{ks}\|_2 \|G_{s\ell}\|_2.
\end{equation}
Let $\psi^y \in \ell^2(X)$ such that  $\psi^y_s = e^{-\tilde \gamma |s-y|}$, then \eqref{e:temp1} can be rewritten as
\begin{equation}
x_k \leq \kappa^{-1}\tbeta(\psi^0_k + \psi^{2N}_k)  +  ( Vx)_k,
\end{equation}
{or, expressed as a vector inequality (understood componentwise),}
\begin{equation}
\label{e:recurs}
x \leq \kappa^{-1}\tbeta(\psi^0 + \psi^{2N})  +  Vx.
\end{equation}
We also have that $\|V\|_2 \leq \beta$ and $\|x\|_2 \leq \kappa^{-1}$. We can then apply \eqref{e:recurs} to itself $m$-times to obtain
\begin{equation*}
x \leq \kappa^{-1}\tbeta \sum_{j=0}^mV^j(\psi^0+\psi^{2N}) +  V^mx.
\end{equation*}
Next, we use the exponential localization of $V$, which gives $\|[({z'}-V)^{-1}]_{0k}\|_2 \lesssim e^{-\gamma_0 k}$ for $d(z',[-\|V\|_2,\|V\|_2]) \sim \alpha \|V\|_2$ for some $\alpha < 1$ { and $\gamma_0 < \gamma^*$ (Lemma 2.2 of \cite{ChenOrtnerTB})}. We choose a contour $\C$ around $[-\|V\|_2,\|V\|_2]$ such that $d(z',[-\|V\|_2,\|V\|_2]) \sim \alpha \|V\|_2$. Then we have
\begin{equation*}
|[V^j]_{ks}| \lesssim  \frac{1}{2\pi i} \oint_{\C} |z'|^j \bigl| [(z'-V)^{-1}]_{ks}\bigr|\cdot |dz'| \lesssim \alpha_0^je^{-\gamma_0|k-s|}.
\end{equation*}
{Here
\begin{equation*}
\alpha_0 = \sup_{z' \in \C} d(z', [-\|V\|_2,\|V\|_2]).
\end{equation*}}
This then gives
\begin{equation*}
| \sum_s [V^j]_{ks}e^{-{ \gamma^*} s}| \lesssim \alpha_0^j\sum_{s \geq 0} e^{-\gamma_0 |k-s|}e^{-\gamma^* s} \lesssim k\alpha_0^je^{- \gamma_0 k},
\end{equation*}
hence we deduce
\begin{equation*}
x_k \lesssim \kappa^{-1} k  {\alpha_0^m}e^{-\gamma_0 k} + \beta^m \kappa^{-1}.
\end{equation*}
Choosing $m \sim k$ with appropriate balancing constant we obtain
\begin{equation*}
x_k \lesssim \kappa^{-1}e^{-\gamma' k},
\end{equation*}
{for some $\gamma'$ dependent on $\eta$, $\beta$, and $\gamma^*$}, which establishes the desired result. {Note that for $\beta < \beta_0$ for some $\beta_0 < 1$, we can guarantee $\gamma' > 0$.}
\end{proof}

We define a contour $\C$ enclosing the spectrum of $\wH_n$ $\forall n$ such that $\Imag(z) = \pm\kappa$ on $\alpha[-\|\wH_n\|_2,\|\wH_n\|_2]$ for some $\alpha$ such that $ \kappa \ll \alpha-1$ (See Figure \ref{fig:contour}). Then,
\begin{equation*}
\begin{split}
& \oJ_0 [\phi_{\epsilon\kappa} \circ \wH_{2N} - \phi_{\epsilon\kappa}\circ \wH_n] \oJ_0^* \\
&= \frac{1}{2\pi i}\oJ_0 \oint_{\C}\phi_{\epsilon\kappa}(z)[ (zI- \wH_{{2N}})^{-1} - (zI- \wH_n)^{-1}]dz \oJ_0^*\\
&= \frac{1}{2\pi i}\oJ_0 \oint_{\C}\phi_{\epsilon\kappa}(z)[ (zI- \wH_{2N})^{-1}(\wH_{2N}-\wH_n) (zI- \wH_n)^{-1}]dz \oJ_0^*\\
&= \frac{1}{2\pi i}\oint_{\C}\phi_{\epsilon\kappa}(z)[ \oJ_0(zI- \wH_{2N})^{-1}\sum_{0 \leq k \leq 2N}\oJ_k^*\oJ_k(\wH_{2N}-\wH_n)\sum_{0 \leq s \leq 2N}\oJ_s^*\oJ_s (zI- \wH_n)^{-1}]dz \oJ_0^*\\
\end{split}
\end{equation*}
We pick $L >0$ large but small enough such that we have {$\kappa L < \frac{1}{2}(\beta^{-1}-1)$}.
Therefore we have
\begin{equation*}
\begin{split}
\|\oJ_0 [\phi_{\epsilon\kappa} \circ \wH_{2N} - \phi_{\epsilon\kappa}\circ \wH_n] \oJ_0^*\|_2 \lesssim &\oint_{\C : |\Real(z)-\epsilon| < L\kappa} |\phi_{\epsilon \kappa}(z)| \sum_{0 \leq k \leq n < s\leq 2N} \kappa^{-2} e^{-\gamma' (k+s)}  e^{- {\gamma^*} |k-s|}\\
& + \kappa^{-3} e^{-L^2/2}.
\end{split}
\end{equation*}
Hence there exists $\tilde \gamma' > 0$ such that
\begin{equation*}
\|\oJ_0 [\phi_{\epsilon\kappa} \circ \wH_{2N} - \phi_{\epsilon\kappa}\circ \wH_n] \oJ_0^*\|_2 \lesssim \kappa^{-3} e^{-\tilde \gamma'n} + \kappa^{-3} e^{-L^2/2}.
\end{equation*}
We can pick $L \sim n^{1/2}$, and so we conclude there is a $\gamma>0$ such that
\begin{equation*}
\|\oJ_0 [\phi_{\epsilon\kappa} \circ \wH_{2N} - \phi_{\epsilon\kappa}\circ \wH_n] \oJ_0^*\|_2 \lesssim \kappa^{-3} e^{-\gamma r}.
\end{equation*}
This completes the proof of Theorem \ref{thm:approx}.

\begin{figure}[ht]
\centering
\includegraphics[width=.5\linewidth]{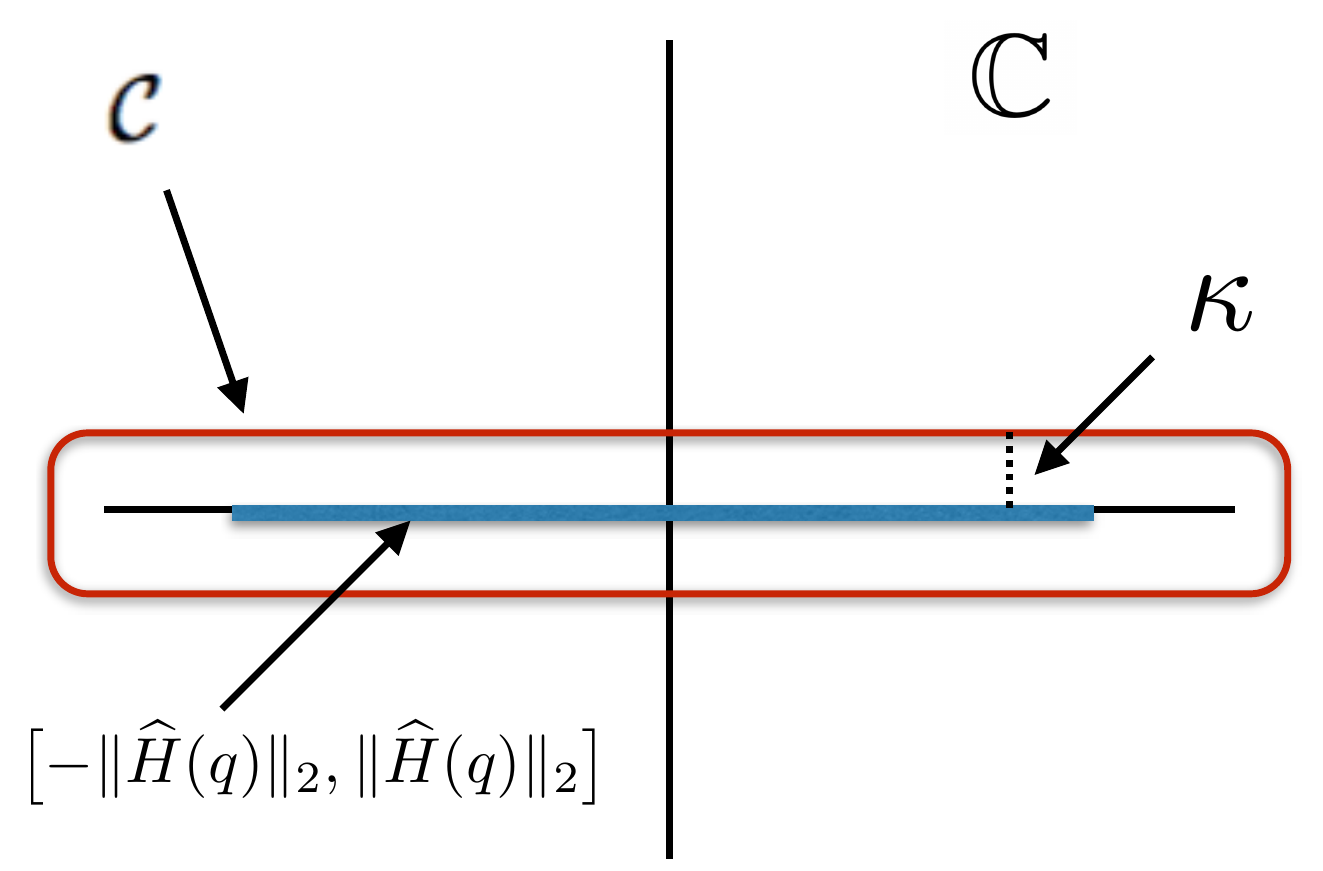}
\caption{Contour enclosing the spectrum of $\widehat{H}_n$ $\forall n$.}
\label{fig:contour}
\end{figure}

\end{section}

\end{document}